\documentclass[a4paper,fleqn]{article}

\usepackage[utf8]{inputenc}
\usepackage{algorithm}
\usepackage[noend]{algorithmic}
\usepackage{amssymb}
\usepackage{amsmath}

\usepackage{tikz}
\usetikzlibrary{decorations.pathreplacing}
\usetikzlibrary{decorations.pathmorphing}
\usetikzlibrary{shapes,shapes.symbols,automata,arrows}
\usetikzlibrary{calc}
\usetikzlibrary{patterns}

\tikzset{p0/.style = {shape = circle, draw, thick, minimum size = 0.7cm}}
\tikzset{p1/.style = {rectangle, minimum size=.7cm, draw, thick}}
\tikzset{p0s/.style = {shape = circle, draw, thick, minimum size = 0.4cm}}
\tikzset{p1s/.style = {rectangle, minimum size=.4cm, draw, thick}}
\tikzset{>=stealth, shorten >=1pt}
\tikzset{every edge/.style = {thick, ->, draw}}
\tikzset{every loop/.style = {thick, ->, draw}}
\tikzstyle{fault}=[dashed]

\newcommand{\myquot}[1]{``#1''}

\newcommand{\bigo}[0]{\mathcal{O}}

\newcommand{\size}[1]{|#1|}

\newcommand{\set}[1]{\{ #1 \}}
\renewcommand{\epsilon}{\varepsilon}

\newcommand{\arena}{\mathcal{A}}
\newcommand{\game}{\mathcal{G}}
\newcommand{\col}{\Omega}
\newcommand{\wincond}{\mathrm{Win}}
\newcommand{\winreg}{\mathcal{W}}

\newcommand{\safety}{\mathrm{Safety}}
\newcommand{\reach}{\mathrm{Reach}}
\newcommand{\buechi}{\text{Büchi}}
\newcommand{\parity}{\mathrm{Parity}}

\newcommand{\exptime}{\textsc{ExpTime}}
\newcommand{\expspace}{\textsc{ExpSpace}}

\newcommand{\twoexp}{\textsc{2ExpTime}}
\newcommand{\pspace}{\textsc{PSpace}}

\newcommand{\disturbances}[0]{\#_D}
\newcommand{\rig}{\mathrm{rig}}

\newcommand{\Gammabot}{\Gamma_{\!\!\bot}}
\newcommand{\initmark}{I}
\newcommand{\pds}{\mathcal{P}}
\newcommand{\Epsilon}{\mathcal{E}}
\newcommand{\sh}{\mathrm{sh}}
\newcommand{\prim}[1]{p_{#1}\#}
\newcommand{\val}{\mathrm{val}}
\newcommand{\rep}{\mathrm{rep}}
\newcommand{\maxheight}{\mathrm{maxSh}}
\newcommand{\opt}{\mathrm{opt}}

\newcommand{\att}[0]{\mathrm{Att}_1}
\newcommand{\bndr}[0]{\mathrm{Bnd}_D}

\newcommand{\sgmark}{\circ}
\newcommand{\sgrankr}{\mu_r}
\newcommand{\sgrankd}{\mu_d}

\usepackage[]{amsthm}
\newtheorem{theorem}{Theorem}
\newtheorem{remark}[theorem]{Remark}
\newtheorem{example}[theorem]{Example}
\newtheorem{lemma}[theorem]{Lemma}
\newtheorem{proposition}[theorem]{Proposition}
\usepackage[]{url}
\usepackage{thmtools,thm-restate}

\usepackage{fullpage}
\usepackage[hidelinks]{hyperref}
\newcommand*{\email}[1]{\href{mailto:#1}{\nolinkurl{#1}} }

\title{Optimally Resilient Strategies in Pushdown Safety Games}
\author{\textbf{Daniel Neider}\\
Max Planck Institute for Software Systems, 67663 Kaiserslautern, Germany\\ \email{neider@mpi-sws.org}\bigskip \\ 
\textbf{Patrick Totzke}\\ University of Liverpool, Liverpool L69 3BX, United Kingdom\\ \email{totzke@liverpool.ac.uk}\bigskip\\ 
\textbf{Martin Zimmermann}\\ University of Liverpool, Liverpool L69 3BX, United Kingdom\\ \email{martin.zimmermann@liverpool.ac.uk}
}

\date{\vspace{-4ex}}

\begin{document}

\maketitle

\begin{abstract}
Infinite-duration games with disturbances extend the classical framework of infinite-duration games, which captures the reactive synthesis problem, with a discrete measure of resilience against non-antagonistic external influence.
This concerns events where the observed system behavior differs from the intended one prescribed by the controller.
For games played on finite arenas it is known that computing optimally resilient strategies only incurs a polynomial overhead over solving classical games.

This paper studies safety games with disturbances played on infinite arenas induced by pushdown systems.
We show how to compute optimally resilient strategies in triply-expo\-nential time.
For the subclass of safety games played on one-counter configuration graphs, we show that determining the degree of resilience of the initial configuration is PSPACE-complete and that optimally resilient strategies can be computed in doubly-exponential time.
\end{abstract}

\section{Introduction}
Infinite games on finite arenas are a popular approach to the synthesis of reactive controllers from logical specifications.
Originally proposed by Büchi and Landweber in 1969~\cite{buechi}, many variations of this classical framework have been studied, including stochastic games~\cite{Condon/93/onalgorithms}, games with partial information~\cite{DoyenRaskin11}, games with delays~\cite{HoschLandweber72}, and games over infinite arenas such as pushdown graphs~\cite{Walukiewicz01} and automatic structures~\cite{DBLP:conf/wia/Neider10,DBLP:conf/tacas/NeiderT16}.
Other variations of this framework stem from the desire to synthesize controllers that exhibit certain user-desired properties.
Examples of such properties range from controllers that need to achieve their task, e.g., reaching a goal, as quickly as possible~\cite{Cachat02} to controllers that are ``robust'' or ``resilient'' with respect to the environment in which they are deployed~\cite{DBLP:journals/acta/BloemCGHHJKK14,DBLP:journals/tcs/MullerS85,DBLP:conf/cav/BloemCHJ09,DBLP:journals/tse/HuangPSW16,DBLP:journals/tecs/MajumdarRT13,DBLP:journals/tac/TabuadaCRM14,DBLP:conf/csl/TabuadaN16,DBLP:conf/hybrid/TopcuOLM12}.
Furthermore, infinite games have a plethora of applications in logic, automata theory and verification beyond the synthesis of reactive controllers.
In this paper, we are concerned with the synthesis application and study infinite games with so-called \emph{unmodeled intermittent disturbances}~\cite{DBLP:conf/cdc/DallalNT16} played on configuration graphs of pushdown machines (pushdown graphs). 

Pushdown graphs are finitely represented infinite graphs, typically the simplest  class of such graphs one studies. 
Despite being conceptually simple, they have natural applications in program analysis, static code analysis, and compiler optimization~\cite{RepsHS95,RepsLK07} due to their ability to capture recursion, e.g., the call stack of a procedural program.
Furthermore, pushdown graphs are known to be well-behaved, and many problems on pushdown graphs are decidable (see, e.g.,~\cite{DBLP:journals/jcss/BohmGJ14,DBLP:journals/tcs/Senizergues01,DBLP:journals/siamcomp/Senizergues05,srba}). 
In particular, Walukiewicz showed that solving parity games played on pushdown graphs is $\exptime$-complete~\cite{Walukiewicz01}, pav\-ing the way for effective synthesis of \emph{recursive} controllers. 
Also, Walukiewicz's result started a long and fruitful line of work on games on pushdown graphs~\cite{Cachat02,DBLP:conf/icalp/Cachat03,CarayolH18,KV2000,Serre06}.
Of particular interest is the special case of games on configuration graphs of one-counter machines, i.e., pushdown machines with a single stack symbol, which is known to be $\pspace$-complete~\cite{Serre06,JancarS07}.

Games with unmodeled intermittent disturbances were originally introduced by Dallal, Neider, and Tabuada~\cite{DBLP:conf/cdc/DallalNT16} to synthesize resilient controllers. 
The observation underlying this type of infinite game is that modeling the real-world environment of a controller in sufficiently great detail is often extremely challenging, either because parts of the environment are unknown or because simulating the environment is costly. 
Moreover, even if a high-resolution model of the environment is available, the resulting games often become prohibitively large.
To alleviate this serious obstacle, Dallal, Neider, and Tabuada proposed to augment classical games with what they call unmodeled intermittent disturbances (in the following just called \emph{disturbances} for the sake of brevity).
Intuitively, such disturbances modify the outcome of a control action, thus modeling that the intended action of the controller did not have the desired consequences.
Note, however, that disturbances are not under the control of the environment and, thus, are not antagonistic.
Similarly, one does not consider the occurrence of disturbances as random  events, as coming up with an appropriate stochastic error model is typically hard. 
Instead, the reader should understand them as rare events, such as a robot arm failing to grab an object due a physical phenomenon that has not been fully modeled.

The original work of Dallal, Neider, and Tabuada~\cite{DBLP:conf/cdc/DallalNT16} provides a method to compute \emph{optimally resilient strategies} for safety games over finite arenas,
which intuitively are winning strategies that can tolerate as many disturbances as possible.
In follow-up work, Neider, Weinert, and Zimmermann~\cite{NeiderW018} have shown that computing optimally resilient strategies in finite arenas only incurs a polynomial overhead over solving classical games (under some mild assumptions on the winning condition), i.e., whenever a class of games is solvable without disturbances, then it is also solvable with disturbances.
In particular, they have developed an algorithm that is effective for all standard winning conditions such as Rabin, Muller, and parity.
Note, however, that both approaches crucially rely on the arena being finite.

 The natural question, which we address here, is how to compute optimally resilient strategies for games on infinite arenas.
As this is a very ambitious goal in its full generality, we restrict ourselves here to the setting
of \emph{safety games} played on \emph{pushdown graphs}.\footnote{
Some of our results do carry over to other winning conditions, such as reachability and parity,  or do not require the underlying arena to be a pushdown graph.
If this is the case, we present our arguments and state our results as general as possible.
Also, we discuss the additional challenges one has to overcome to generalize all our results to reachability and parity conditions.}

As argued before, pushdown games are a natural starting point for investigating effective algorithms for games on infinite graphs,
and safety specifications are a fundamental class of specifications in practice~\cite{DBLP:conf/icse/DwyerAC99}. 
While this setting might seem restrictive, recall that both the $\exptime$-hard\-ness of solving pushdown games~\cite{Walukiewicz01} and the $\pspace$-hard\-ness of solving one-counter games~\cite{Serre06} already hold for the safety condition. 
Thus, the complexity of solving pushdown games stems from the transition from finite to infinite graphs, not from the expressiveness of the winning condition.
The setting we consider here is still  expressive enough to model interesting applications such as reasoning about exception handling in recursive programs.
Here, one is interested in determining how many exceptions the program can tolerate while still satisfying a given specification. 

To capture the optimization aspect of the problem at hand, we re-use Neider, Weinert, and Zimmermann's notion of \emph{resilience values}~\cite{NeiderW018}, which assigns to every vertex~$v$ of the arena an ordinal $r_\game(v) \leq \omega + 1$, where $\game$ denotes the game in question and $\omega$ is the first infinite ordinal.
Intuitively, $r_\game(v)$ denotes how many disturbances can be tolerated by an optimally resilient strategy from $v$. This value can be $k \in \omega$ ($k-1$ disturbances can be tolerated, but not $k$), $\omega$ (finitely many disturbances can be tolerated, but not infinitely many), or $\omega+1$ (infinitely many disturbances can be tolerated).
When moving from finite to infinite arenas, however, various conceptual and technical complications arise, which make computing the resilience values of vertices and, by extension, resilient strategies challenging.

For instance, safety games over infinite arenas no longer guarantee the existence of optimally resilient strategies, i.e., in an infinite arena, one does not necessarily have a strategy that can tolerate an arbitrary finite number of disturbances from a vertex with resilience $\omega$.
Instead one has, for every $k \in \omega$, a strategy that can tolerate $k$ disturbances, but not $k+1$.

Another complication is the fact that it is no longer possible to globally bound the finite resilience values in infinite arenas. 
In contrast, in the case of finite arenas, the number of vertices is a trivial bound on the finite resilience values~\cite{NeiderW018}.
Hence, fixed-point algorithms like the ones devised for finite arenas~\cite{DBLP:conf/cdc/DallalNT16,NeiderW018} and algorithms based on exhaustive search do not necessarily terminate.

\subsection*{Our Contributions}
In the rest of this paper, we study resilience in pushdown safety games, which we introduce in Section~\ref{sec_prelims}.

First, we show in Section~\ref{sec_arbitrary} that 
no vertex of a finitely branching safety game (which covers pushdown games in particular) can have resilience~$\omega$.
As a corollary,
we show that Player~$0$ has positional optimally resilient strategies in finitely branching safety games.
In contrast, we show that Player~$0$ does not necessarily have an optimally resilient strategy in infinitely branching safety games, for the reasons explained earlier.

In Sections~\ref{sec_riggedgames} to \ref{sec_ocs}, we consider the problem of determining the resilience of the initial vertex of a given pushdown safety game. 
First, we show in Section~\ref{sec_riggedgames} how to characterize resilience values using classical games (without disturbances):
While the notion of resilience is not defined via strategies of the antagonist, we show that one can nevertheless give control over disturbances to the antagonist, if one additionally adjusts the winning condition to control the number of occurrences of disturbances. 
For certain resilience values, but not all, this adjustment leads to a polynomial time reduction to solving classical games
on pushdown games.
The values that can be characterized in safety games are fixed finite values~$k$ and $\omega+1$, but not $\omega$. 

We then prove that the resilience value of the initial vertex in pushdown safety games can determined in triply-exponential time (Sections~\ref{sec_pushdown}) and that of the initial vertex in one-counter safety games in polynomial space (Section~\ref{sec_ocs}).
The latter result is tight, as associated decision problems are shown to be $\pspace$-complete.
To show membership, we use the following approach: 
We prove the existence of an upper bound on the resilience value of the initial vertex in case it is finite. 
With such an upper bound~$b$, we can use the characterizations developed in Section~\ref{sec_riggedgames} to perform an exhaustive search on the finite search space (the resilience is either in $\set{0,1,\ldots, b}$ or $\omega+1$, as we have ruled out $\omega$).
For general pushdown games, this search can be implemented in triply-exponential time, as the bound~$b$ is doubly-exponential.
However, relying on the simplicity of configuration graphs of one-counter systems and on the fact that the bound~$b$ is only exponential in this case, we are able to show that the search can be implemented in polynomial space for one-counter safety games.
Proving the last result requires the combination of a wide range of techniques, including results from the theory of quantitative pushdown games~\cite{FridmanZ12}, positional determinacy for quantitative pushdown games, and specifically tailored ``hill-cutting'' \cite{DBLP:journals/jcss/BohmGJ14,Val1973} and ``summarization'' arguments \cite{RepsHS95,HMM2016}, which we generalize from individual paths in pushdown systems to  strategies.
Also, we show that a strategy that is optimally-resilient from the initial vertex can be computed in exponential space (triply-exponential time) for one-counter safety games (pushdown safety games).

Section~\ref{sec_conclusion} concludes and discusses directions for future work.
Finally, we present an application of our results, namely, a connection between optimally resilient strategies in pushdown safety games and 
optimal strategies (in the number of steps to the target) in pushdown reachability games~\cite{Cachat02,CarayolH18}.
There, we also discuss which of our results obtained here carry over to pushdown reachability games and discuss the obstacles preventing us from generalizing the other results from safety to reachability.

\subsection*{Related Work}
\label{subsec_relatedwork}
Resilience, and closely related notions like fault-tolerance and robustness, are not a novel concept in the context of reactive systems synthesis, with numerous formalizations having been proposed.
So as to not clutter this paper too much, we refer the reader to Dallal, Neider, and Tabuada~\cite{DBLP:conf/cdc/DallalNT16} as well as Neider, Weinert, and Zimmermann~\cite{NeiderW018} for a comprehensive discussion of how these notions are related to the concept of unmodeled intermittent disturbances.
Other notions of resilience against environmental impacts not discussed there include  an approach
based on imperfect information games that quantifies the resilience of controllers to noise in the input signal~\cite{Almagor2017,10.1007/978-3-642-19805-2_19} (see also the references).

Finally, let us mention that one can implement the characterization of finite resilience values presented in Section~\ref{sec_riggedgames} by energy conditions~\cite{DBLP:conf/formats/BouyerFLMS08,DBLP:conf/emsoft/ChakrabartiAHS03}.
However, solving energy games on pushdown graphs is undecidable~\cite{AbdullaAHMKT14} and so we do not pursue this approach here.
Similarly unfeasible are stochastic methods to quantify resilience in pushdown games. Indeed,
checking even the most basic, almost-sure reachability conditions for stochastic games on pushdown graphs
is undecidable already for single state systems or single-player games \cite{PMDPsundec}.

\section{Preliminaries}
\label{sec_prelims}
We use the ordinals~$0 < 1< 2 < \cdots < \omega < \omega+1 < \omega+2$ to define resilience values. 
For convenience of notation, we also denote the cardinality of $\omega$ by $\omega$.  

\subsection{Infinite Games with Disturbances}
\label{subsec_games}
An \emph{arena} (with unmodeled intermittent disturbances)~$\arena = (V, V_0, V_1, E, D)$ consists of a countable directed graph~$(V, E)$, a partition~$\set{V_0, V_1}$ of $V$ into the set of vertices~$V_0$ of Player~$0$  and the set of vertices~$V_1$ of Player~$1$, and a set~$D \subseteq V_0 \times V$ of disturbance edges. 
Note that only vertices of Player~$0$ may have outgoing disturbance edges.
We require that every vertex~$v \in V$ has a successor $v'$ with $(v,v') \in E$ to avoid finite plays. A vertex~$v \in V$ is a \emph{sink} if it has a single outgoing edge~$(v,v) \in E$ leading back to itself but no outgoing disturbance edges. 

A \emph{play} in $\arena$ is an infinite sequence~$\rho = (v_0, b_0) (v_1, b_1) (v_2, b_2) \cdots \in (V\times\set{0,1})^\omega$
 such that $b_0 = 0$ and for all $j>0$: 
$b_j = 0$ implies $(v_{j-1}, v_j) \in E$, and $b_j = 1$ implies $(v_{j-1}, v_j) \in D$. 
Hence, the additional bits~$b_j$ for~$j > 0$ denote whether a standard edge or a disturbance edge has been taken to move from $v_{j-1}$ to $v_j$. 
We say $\rho$ starts in $v_0$. 
A play prefix~$(v_0, b_0) \cdots (v_j, b_j)$ is defined similarly and ends in $v_j$. 
The number of disturbances in a play~$\rho = (v_0, b_0) (v_1, b_1) (v_2, b_2) \cdots$ is defined as $\disturbances(\rho) = \size{\set{j \in \omega \mid b_j = 1}}$, which is either some $k \in \omega$ (if there are finitely many disturbances, namely $k$) or it is equal to $\omega$ (if there are infinitely many). 
A play~$\rho$ is disturbance-free, if $\disturbances(\rho) = 0$.

A \emph{game} (with unmodeled intermittent disturbances)~$\game = (\arena, \wincond)$ consists of an arena with set~$V$ of vertices and a winning condition~$\wincond \subseteq V^\omega$. 
 A play~$\rho = (v_0, b_0) (v_1, b_1) (v_2, b_2) \cdots$
  is winning for Player~$0$ if  $v_0 v_1 v_2 \cdots \in \wincond$, otherwise it is winning for Player~$1$. 
 Hence, winning is oblivious to occurrences of disturbances. 

In this work, we focus on safety conditions, but also use the Büchi and parity condition in proofs.
The former two are induced by a subset~$F$ of the set~$V$ of vertices while the latter is induced by a coloring~$\col\colon V \rightarrow \omega$, which is required to have a finite range~$\col(V)$.

\begin{itemize}
	
\item $\safety(F)$ containing the sequences~$v_0 v_1 v_2 \cdots \in V^\omega$ with $ v_j \notin F$  for every $j \in \omega$ denotes the safety condition induced by $F$, which requires to avoid $F$.

\item $\buechi(F)$ containing the sequences~$v_0 v_1 v_2 \cdots \in V^\omega$ with $v_j \in F$  for infinitely  many $j \in \omega$ denotes the Büchi condition induced by $F$, which requires to visit $F$ infinitely often.

\item  $\parity(\col)$  containing the sequences~$ v_0 v_1 v_2 \cdots \in V^\omega$ with even $\limsup \col(v_0) \col(v_1) \col(v_2) \cdots$
denotes the (max-) parity condition induced by $\col$, which requires the maximal color occurring infinitely often during a play to be even. 
As the range of $\col$ is finite, every play has a maximal color occurring infinitely often.

\end{itemize}
A game~$(\arena, \wincond)$ is a \emph{safety game} if $\wincond = \safety(F)$ for some subset~$F$ of the vertices of $\arena$.

A \emph{strategy} for Player~$i \in \set{0,1}$ is a function~$\sigma \colon V^*V_i \rightarrow V$ such that $(v_j, \sigma(v_0 \cdots v_j)) \in E$ for every $v_0 \cdots v_j \in V^*V_i$. 
A play~$(v_0,b_0) (v_1, b_1) (v_2,b_2) \cdots $ is \emph{consistent with} $\sigma$ if $v_{j+1} = \sigma(v_0 \cdots v_j)$ for every $j$ with $v_j \in V_i$ and $b_{j+1} = 0$, i.e., if the next vertex is the one prescribed by the strategy unless a disturbance edge is used. 
A strategy~$\sigma$ is positional, if $\sigma(v_0 \cdots v_j) = \sigma(v_j)$ for all $v_0 \cdots v_j \in V^*V_i$. 

\subsection{Pushdown Games}
\label{subsec_pushdowndefs}
A \emph{pushdown system} (PDS)~$\pds = (Q, \Gamma, \Epsilon,  q_{\initmark})$ consists of a finite set~$Q$ of states  with an initial state~$q_{\initmark}\in Q$, a stack alphabet~$\Gamma$ with a designated stack bottom symbol~$\bot\notin\Gamma$, and a transition relation~$\Epsilon\subseteq Q \times \Gammabot \times Q \times \Gammabot^{\leq 2}$, where $\Gammabot=\Gamma\cup\{\bot\}$ and $\Gammabot^{\le 2} = \set{w \in \Gammabot^* \mid \size{w} \le 2}$. 
We require $\Epsilon$ to neither write nor delete $\bot$ from the stack.
Also, we assume every PDS to be deadlock-free, i.e., for every $q\in Q$ and $A\in\Gammabot$ there exist $q'\in Q$ and $w\in\Gammabot^{\leq 2}$ such that $(q,A,q',w)\in\Epsilon$.
Finally, $\pds$ is a \emph{one-counter system} (OCS) if $\size{\Gamma} = 1$. 

A stack content is a word in $\Gamma^*\bot$ where the leftmost symbol is assumed to be the top of the stack. 
A \emph{configuration} of $\pds$ is a pair $(q,\gamma)$ consisting of a state $q\in Q$ and a stack content~$\gamma\in\Gamma^*\bot$. 
The stack height of a configuration $(q,\gamma)$ is defined by $\sh(q,\gamma)=\size{\gamma}-1$. 
Given two configurations~$(q,\gamma)$ and $(q',\gamma')$ we write $(q,\gamma)\vdash_\Epsilon(q',\gamma')$ if there exists a transition~$(q,\gamma_0,q',w)\in\Epsilon$ with $\gamma'=w\gamma_1\cdots \gamma_{\size{\gamma}-1}$.

Fix a PDS~$\pds = (Q, \Gamma, \Epsilon,  q_{\initmark})$, a partition~$\set{Q_0, Q_1}$ of $Q$ and an additional transition relation $\Delta \subseteq Q_0 \times \Gammabot \times Q \times \Gammabot^{\leq 2}$, which is also required to neither write nor delete $\bot$ from the stack. 
These induce the (pushdown) arena~$(V, V_0, V_1, E, D)$ with
\begin{itemize}
	\item $V=\set{(q,\gamma)\mid q\in Q, \gamma\in\Gamma^*\bot}$ is the set of configurations of $\pds$,
	\item $V_i = \set{(q,\gamma)\in V\mid q\in Q_i}$ for $i \in \set{0,1}$ is the set of configurations whose state is in $Q_i$,
	\item $E = \set{(v,v') \mid v\vdash_\Epsilon v'}$ is the set of edges, induced by the transition relation~$\Epsilon$, and
	\item $D = \set{(v,v') \mid v\vdash_\Delta v'}$ is the set of disturbance edges, which is induced by the transition relation~$\Delta$, where $\vdash_{\Delta}$ is defined analogously to $\vdash_\Epsilon$.
\end{itemize}
Typically, we are interested in the initial vertex of the arena, which is defined as $(q_\initmark, \bot)$.

A \emph{pushdown safety game} is a safety game whose arena is induced by a pushdown system~$\pds$ and whose winning condition is induced by a subset of $\pds$'s states, i.e., $F \subseteq Q$ induces the set~$\set{(q, \gamma) \in V \mid q \in F} $ of vertices.
One-counter safety games are defined analogously.

When using a pushdown game as an input for an algorithm, we represent it by the underlying PDS, the partition of its states, the additional transition relation for the disturbance edges, and a subset of the states inducing the winning condition. 
We define the size of the input as $\size{Q} + \size{\Gamma}$, as all these objects can be represented in polynomial size in the number of states and stack symbols of the underlying PDS.

\subsection{Infinite Games without Disturbances}
\label{subsec_gameswithoutdisturbances}
For technical convenience, we characterize the classical notion of infinite games, i.e., those without disturbances, (see, e.g.,~\cite{GraedelThomasWilke02}) as a  special case of games with disturbances. 
Let $\game$ be a game with vertex set~$V$. 
A strategy~$\sigma$ for Player~$i$ in $\game$ is said to be a winning strategy for her from $v \in V$, if every disturbance-free play that starts in $v$ and that is consistent with $\sigma$ is winning for Player~$i$.
The winning region~$\winreg_i(\game)$ of Player~$i$ in $\game$ contains those vertices from which Player~$i$ has a winning strategy.
Thus, the winning regions of $\game$ are independent of the disturbance edges, i.e., we obtain the classical notion of infinite games.
Player~$i$ wins $\game$ from $v$, if $v \in \winreg_i(\game)$. 

\subsection{Resilient Strategies}
\label{subsec_gameswithdisturbances}
Let $\game$ be a game with vertex set~$V$ and let $\alpha \in \omega+2$. 
A strategy~$\sigma$ for Player~$0$ in $\game$ is \emph{$\alpha$-resilient} from~$v \in V$ if every play~$\rho$ that starts in $v$, that is consistent with $\sigma$, and with~$\disturbances(\rho) < \alpha$, is winning for Player~$0$. 
Thus, a $k$-resilient strategy with $k \in \omega$ is winning even under at most $k-1$ disturbances, an $\omega$-resilient strategy is winning even under any finite number of disturbances, and an $(\omega+1)$-resilient strategy is winning even under infinitely many disturbances. 

We define the \emph{resilience} of a vertex~$v$ of $\game$ as 
\[
r_\game(v) = \sup\set{ \alpha\in\omega+2  \mid \text{Player~$0$ has an} \text{$\alpha$-resilient strategy for $\game$ from $v$}}.
	\]
Note that the definition is not antagonistic, i.e., it is not defined via strategies of Player~$1$. 
A strategy~$\sigma$ is \emph{optimally resilient} if it is $r_\game(v)$-resilient from every vertex~$v$.

\begin{example}
\label{example_pdsexampleresilience}
Consider the game~$\game = (\arena, \safety(F))$ where $\arena$ is the arena from Figure~\ref{fig_pdsexample} and $\safety(F)$ is the safety condition induced by $F = \set{q_2}$.
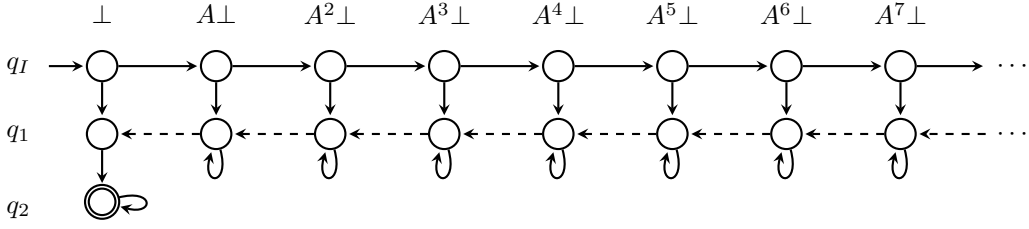
\begin{figure*}[t]
\centering
\begin{tikzpicture}[ultra thick]

\foreach \i in {0,1,...,7}{
\node[p0s] (i\i) at (\i*1.5,-.4) {};
\node[p0s] (1\i) at (\i*1.5,-1.3) {};
}
\node 	  (idots) at (12,-.4) {$\cdots$};
\node 	  (1dots) at (12,-1.3) {$\cdots$};
\node[p0s,double] (20) at (0,-2.2) {};

\foreach \i in {0,...,7}{
\path[-stealth]
(i\i) edge (1\i)
;
}
\foreach \i in {1,...,7}{
\path[-stealth]
(1\i) edge[loop below] ()
;
}

\foreach \i [remember=\i as \lasti (initially 0)]in {1,...,7}{
\path[-stealth]
(i\lasti) edge (i\i)
;
}

\foreach \i [remember=\i as \lasti (initially 7)]in {6,...,0}{
\path[-stealth]
(1\lasti) edge[fault] (1\i)
;
}

\node at (-1.1, -.4) {$q_\initmark$};
\node at (-1.1, -1.3) {$q_1$};
\node at (-1.1, -2.3) {$q_2$};

\node at (0, .3) {$\bot$};
\node at (1.5, .3) {$A\bot$};
\node at (3, .3) {$A^2\bot$};
\node at (4.5, .3) {$A^3\bot$};
\node at (6, .3) {$A^4\bot$};
\node at (7.5, .3) {$A^5\bot$};
\node at (9, .3) {$A^6\bot$};
\node at (10.5, .3) {$A^7\bot$};

\path[-stealth]
(-.7,-.4) edge (i0)
(i7) edge (idots)
(10) edge (20)
(1dots) edge[fault] (17)
(20) edge[loop right] ()
;

\end{tikzpicture}
\caption{A one-counter arena, restricted to vertices reachable from the initial vertex~$(q_\initmark, \bot)$. All vertices are in $V_0$, disturbance edges are drawn as dashed arrows, and doubly-lined vertices are in $F$.}
\label{fig_pdsexample}
\end{figure*}

We have that $r_\game(q_\initmark, A^n\bot) = \omega+1$, $r_\game(q_1,A^n\bot) = n$ for all $n \in \omega$, and $r_\game(q_2, \bot) = 0$.
Furthermore, the strategy that indefinitely stays in state~$q_\initmark$ is optimally resilient.
\end{example}

Next, we state some simple properties of resilient strategies and resilience values that are useful throughout the paper.

\begin{remark}
\label{remark_reconstructingdisturbances}
A strategy~$\sigma$ does not have access to the bits indicating whether a disturbance occurred or not. 
However, this is not a restriction: 
let $(v_0,b_0) (v_1, b_1) (v_2,b_2) \cdots $ be a play with $b_j = 1$ for some $j > 0$. 
We say that this disturbance is consequential (w.r.t.\ $\sigma$), if $v_j \neq \sigma(v_0 \cdots v_{j-1})$, i.e., if the disturbance transition~$(v_{j-1}, v_j)$ traversed by the play did not lead to the vertex the strategy prescribed. 
Such consequential disturbances can be detected by comparing the actual vertex~$v_j$ to $\sigma$'s output~$\sigma(v_0 \cdots v_{j-1})$. 
On the other hand, inconsequential disturbances will just be ignored. 
In particular, the number of consequential disturbances is always at most the number of disturbances.
\end{remark}

The following remark lists some simple consequences of the definition of resilience.

\begin{remark}
\label{remark_resilienceproperties}
The following hold for every vertex $v$.
\begin{enumerate}

	\item\label{remark_resilienceproperties_zeroresilience} Every strategy is $0$-resilient from $v$.

	\item\label{remark_resilienceproperties_winningstrategyplzero} A strategy is~$1$-resilient from~$v$ if and only if it is winning for Player~$0$ from~$v$. 

	\item\label{remark_resilienceproperties_monotonicity} 
            If a strategy is $\alpha$-resilient from $v$
            and 
            $\alpha > \alpha'$
            then it is also $\alpha'$-resilient from $v$.

\end{enumerate}
\end{remark}

Note that every game has disjoint winning regions. 
A game is determined, if every vertex is in either winning region.
The previous remark implies that resilience refines winning regions.

\begin{lemma}
\label{lemma_resiliencegeneralizeswinningregions}
Let $\game$ be a game and $v$ a vertex of $\game$.
\begin{enumerate}
\item\label{lemma_resiliencegeneralizeswinningregions_playerzero} $r_\game(v) > 0$ if and only if $v \in \winreg_0(\game)$.
\item\label{lemma_resiliencegeneralizeswinningregions_playerone} If $\game$ is determined, then $r_\game(v) = 0$ if and only if $v \in \winreg_1(\game)$.
\end{enumerate}
\end{lemma}

\begin{proof}
\ref{lemma_resiliencegeneralizeswinningregions_playerzero}.) 
The resilience of $v$ is greater than zero if and only if Player~$0$ has a $1$-resilient strategy from $v$ due to Item~\ref{remark_resilienceproperties_monotonicity} of Remark~\ref{remark_resilienceproperties}. 
The latter condition is equivalent to Player~$0$ having a winning strategy for $\game$ from $v$, i.e., equivalent to $v \in \winreg_0(\game)$, due to Item~\ref{remark_resilienceproperties_winningstrategyplzero} of Remark~\ref{remark_resilienceproperties}.

\ref{lemma_resiliencegeneralizeswinningregions_playerone}.) 
Due to Items~\ref{remark_resilienceproperties_monotonicity} and \ref{remark_resilienceproperties_winningstrategyplzero} of Remark~\ref{remark_resilienceproperties}, the resilience of $v$ is zero if and only if Player~$0$ has no winning strategy for $\game$ from $v$, i.e., $v \notin \winreg_0(\game)$. 
Due to determinacy, this is equivalent to $v \in \winreg_1(\game)$.	
\end{proof}

Note that determinacy is a necessary condition for Item~\ref{lemma_resiliencegeneralizeswinningregions_playerone}. In an undetermined game, the vertices that are in neither winning region have resilience zero, due to Item~\ref{lemma_resiliencegeneralizeswinningregions_playerzero}, but are in particular not in $\winreg_1(\game)$.

For determined disturbance-free games, i.e., those without disturbance edges in the arena, we obtain a tighter connection between resilience and winning regions:
There are only two possible resilience values and they characterize the winning regions.

\begin{remark}
\label{remark_resilienceequalswinningregionsdisturbancefreegames}
Let $\game$ be a determined disturbance-free game and $v$ a vertex of $\game$.
\begin{enumerate}

	\item\label{remark_resilienceequalswinningregionsdisturbancefreegames_playerzero}
	 $r_\game(v) = \omega+1$ if and only if $v \in \winreg_0(\game)$.

	\item\label{remark_resilienceequalswinningregionsdisturbancefreegames_playerone}
	 $r_\game(v) = 0$ if and only if $v \in \winreg_1(\game)$.

\end{enumerate}	
\end{remark}

\label{cons_safety2parity}
Finally, given a safety game~$\game = (\arena, \safety(F))$ with vertex set~$V$ we turn all vertices in $F$ into sinks, obtaining the arena~$\arena'$ with vertex set~$V$. 
Then, we have $r_\game(v) = r_{\game'}(v)$, where $\game' = (\arena', \parity(\col))$ for the coloring~$\col$ mapping vertices in $F$ to $1$ and all other vertices to $2$.
This construction will be useful in some proofs later on.

\section{Resilience in Infinite Safety Games}
\label{sec_arbitrary}
Player~$0$ has optimally resilient strategies in every safety game played in a finite arena~\cite{DBLP:conf/cdc/DallalNT16}. 
In this section, we show that this result also holds for pushdown safety games, but fails for safety games in arbitrary infinite arenas.
We start by observing that in safety games in infinite arenas, vertices with resilience~$\omega$ may exist, unlike in safety games in finite arenas~\cite{DBLP:conf/cdc/DallalNT16}.

\begin{example}
\label{example_safetyallpossibleresiliencevalues}	
Consider the one-counter arena presented in Figure~\ref{fig_pdsexample} with the safety condition induced by $F = \set{q_2}$, i.e., Player~$0$ wins if she avoids visiting a vertex with state~$q_2$.
 As argued in Example~\ref{example_pdsexampleresilience}, the resulting game~$\game$ has vertices of resilience~$\omega+1$ and $k$, for each $k \in \omega$, i.e., all values but $\omega$ are assumed. 
 
Let us add a vertex~$v \in V_0$ to $\game$ with outgoing edges to all vertices of the form~$(q_1, A^n\bot)$ to obtain the game~$\game'$ (which is infinitely branching and therefore no longer a pushdown arena). 
 Let $\sigma_k$, for $k > 0$, be a strategy that moves from $v$ to $(q_1, A^k\bot)$.
We have that $r_{\game'}(v) \ge \omega$, as $\sigma_k$ is $k$-resilient from $v$.
Consider an arbitrary strategy~$\sigma$:
From $v$, it moves to some $(q_1, A^k\bot)$ from which $k$ disturbances force the play into the losing sink. 
Hence, $\sigma$ is not $(k+1)$-resilient and therefore not $\omega$-resilient.
Thus, there is no optimally resilient strategy in $\game'$.
\end{example}

The underlying issue is that $r_\game(v) \ge \omega$ can be witnessed either 
\begin{enumerate}
	\item[(a)] by the existence of a strategy that is $\omega$-resilient from $v$, or
	\item[(b)] by the existence of a family~$(\sigma_k)_{k \in \omega}$ of strategies where each $\sigma_k$ is  $ k $-resilient from $v$, but not $\omega$-resilient from $v$.
\end{enumerate}  
The second case only exists as $\omega$ is a limit ordinal (the only one we consider).
For all $\alpha \neq \omega$, we have that $r_\game(v) = \alpha$ if and only if Player~$0$ has an $\alpha$-resilient strategy from $v$.
The games studied in previous work~\cite{DBLP:conf/cdc/DallalNT16,NeiderW018} only exhibited the former case, as these only considered finite arenas. 
As witnessed in Example~\ref{example_safetyallpossibleresiliencevalues}, this is no longer true in games in infinite arenas.

Note that there is a change of quantifiers between these two cases: by definition, an $\omega$-resilient strategy is $ k $-resilient for every $ k  \in \omega$, i.e., in the former case there is a uniform strategy that is $ k $-resilient for every $ k  \in \omega$. 
In the latter case, for every $ k  \in \omega$, there is a strategy that is $ k $-resilient, but not $\omega$-resilient. 
Hence, in the following, we distinguish between these two cases. 
We say that a vertex~$v$ of a game~$\game$ with $r_\game(v) = \omega$ has a uniform witness\footnote{Note that uniformity here refers to having a single strategy~$\sigma$ that is $k$-resilient from $v$ for every $k$. It is \emph{not} related to the concept of uniform winning strategies, i.e., strategies that are winning from every vertex in a winning region.}, if there is an $\omega$-resilient strategy from $v$.
A game with a vertex of resilience~$\omega$ without a uniform witness has no optimally resilient strategy by definition.

For safety games in infinite arenas, the existence of optimally resilient strategies depends on the branching of the arena.
We say that an arena~$(V, V_0, V_1, E, D)$ is \emph{finitely branching} if the set $\set{v' \mid (v,v') \in E}$ of successors of $v$ is finite for every $v \in V$.
Otherwise, if there is a vertex with infinitely many successors, then the arena is \emph{infinitely branching}.
Note that pushdown arenas are finitely branching.

The following theorem shows that the games presented in Example~\ref{example_safetyallpossibleresiliencevalues} already exhibit all possible resilience values in safety games, and that infinite branching is necessary to obtain a vertex of resilience~$\omega$. 
We formulate the result for arbitrary infinite arenas, as the proof technique we use here does not rely on the arena being a pushdown arena.

\begin{lemma}
\label{lemma_resiliencevaluessafety}
Let $\game$ be a safety game with vertex set~$V$.
\begin{enumerate}
	\item\label{lemma_resiliencevaluessafety_infbranch}
	There is no $v \in V$ with $r_\game(v)=\omega$ that has a uniform witness.
	
	\item\label{lemma_resiliencevaluessafety_finbranch}
	 If $\arena$ is finitely branching, then there is no $v \in V$ with $r_\game(v) = \omega$.
\end{enumerate}\
\end{lemma}

\begin{proof}
\ref{lemma_resiliencevaluessafety_infbranch}.)
Let $\game = (\arena, \safety(F))$. 
Towards a contradiction assume that there is a vertex~$v \in V$ with $r_\game(v) = \omega$ and that there is a strategy~$\sigma$ that is $\omega$-resilient from $v$.
Due to $r_\game(v) < \omega+1$, $\sigma$ is not $(\omega+1)$-resilient from $v$. 
Thus, there is a play~$\rho = (v_0, b_0)(v_1, b_1)(v_2, b_2) \cdots$ that starts in $v$, is consistent with $\sigma$, satisfies $\disturbances(\rho) < \omega+1$ (which is a tautology), and such that $v_0 v_1 v_2 \cdots \notin \safety(F)$, i.e., there is a $j$ such that $v_j \in F$.
Consider a play of the form $\rho' = (v_0, b_0) \cdots (v_j, b_j) \rho''$ that is consistent with $\sigma$ and such that $(v_j, b_j)\rho''$ is disturbance-free.
Such a play exists, as each vertex in $V_0$ has a non-disturbance successor.
The play~$\rho'$ starts in $v$, is consistent with $\sigma$, satisfies $\disturbances(\rho') \le j$, as disturbances can only occur in the prefix~$(v_0, b_0) \cdots (v_j, b_j)$, but violates the safety condition, as $v_j \in F$ is visited by $\rho'$. 
Therefore, $\sigma$ is not $(j+1)$-resilient from $v$, and in particular not $\omega$-resilient from $v$, which contradicts our assumption.

\ref{lemma_resiliencevaluessafety_finbranch}.)
We begin by giving a characterization of the resilience values in finitely branching safety games that will be the basis of both the proof of Lemma~\ref{lemma_resiliencevaluessafety}.\ref{lemma_resiliencevaluessafety_finbranch} and the proof of Theorem~\ref{theorem_optimalstrategiesinfinitelybranchinggames}.
The characterization is a generalization of a similar one for safety games in finite arenas~\cite{DBLP:conf/cdc/DallalNT16}.

Fix a finitely branching safety game~$\game = (\arena, \safety(F))$ with $\arena = ( V, V_0, V_1, E, D )$. 
First, we recall the  attractor construction for Player~$1$.  
Fix a set~$X \subseteq V$.
Let $A_0 =  X$ and define, for every $j \ge 0$, $A_{j+1}$ as follows.
\[
A_{j+1} = A_j \cup\set{ v \in V_0 \mid \text{for all $(v,v') \in E$: $v' \in A_j$}} \cup \set{ v \in V_1 \mid \text{there exists $(v,v') \in E$ with $v' \in A_j$}}  
\]
We call $\att(X) = \bigcup_{j \in \omega} A_j$ the $1$-attractor of $X$ in $\arena$.

By construction, Player~$1$ has a positional strategy~$\tau$ such that every disturbance-free play starting in $\att(X)$ and being consistent with $\tau$ visits $X$ at least once. 
Dually, Player~$0$ has a positional strategy~$\sigma$ such that every disturbance-free play starting in $V \setminus \att(X)$ and being consistent with $\sigma$ never visits $X$. 
We refer to $\tau$ and $\sigma$ as the attractor and trap strategy associated to $\att(X)$.
Finally, we call 
\[
\bndr(X) = \set{v \in V_0\setminus X \mid \text{there exists $(v,v') \in D$ with $v' \in X$}}
\]
the $D$-boundary of $X$, which contains all vertices $v \notin X$ from which a disturbance edge leads into $X$. 

In the following, we alternatingly apply the attractor and the boundary operation starting with the set~$F$ of vertices that Player~$0$ has to avoid in order to win. 
Then, we show that every vertex in the limit has finite resilience while every other vertex has resilience~$\omega + 1$, which completes the proof.

Formally, let $S_0  = \att(F)$ be the $1$-attractor of $F$, $S_{j+1} = \att( S_j\cup \bndr(S_j))$ for every $j \in \omega$, and define $S = \bigcup_{j \in \omega}S_j$. 
Now, for $v \in S$, let $r(v) =  \min \set{j \mid v \in S_j}$ be the index at which $v$ is added to $S$.\label{page_r}

We claim $r_\game(v) \le r(v)$ for every $v \in S$ and $r_\game(v) = \omega +1 $ for every $v \notin S$, which proves our claim.

Fix a  vertex~$v \in S$.
To show $r_\game(v) \le r(v)$,  we need to show for every strategy $\sigma$ for Player~$0$ that there is a play that starts in $v$, is consistent with $\sigma$, has at most $r(v)$ disturbances, and is losing for Player~$0$, i.e., it visits $F$ at least once. 
We fix any strategy $\sigma$ and construct such a play inductively starting with the play prefix~$(v_0, b_0) = (v,0)$.
During the construction, we ensure that the prefix constructed thus far is consistent with $\sigma$ and that it ends in $S$.
Thus, assume we have constructed a play prefix~$w = (v_0, b_0) \cdots (v_j,b_j)$ satisfying the invariant.
To extend it, we distinguish two cases:
\begin{enumerate}

\item Assume  $r(v_j) = 0$, i.e., $v_j \in S_0 = \att(F)$.
Then, consider the unique disturbance-free play~$(v_j,0) \rho$ consistent with $\sigma$ and the attractor strategy for Player~$1$ associated with $\att(F)$.
We extend $w$ by $\rho$ to complete the construction of the desired play. 
The resulting play~$w\rho$ is consistent with $\sigma$ due to our invariant and the choice of $\rho$, and  contains a vertex from $F$. 

\item Assume $r(v_j) > 0$, i.e., $v_j \in S_{r(v_j)} = \att( S_{{r(v_j)}-1}\cup \bndr(S_{{r(v_j)}-1}))$. 
Consider the unique disturbance-free play~$(v_j,0) \rho$ consistent with $\sigma$ and the attractor strategy for Player~$1$ associated with 
\[\att( S_{{r(v_j)}-1}\cup \bndr(S_{{r(v_j)}-1})).\] 
Let $(v_j,0) (v_{j+1},0) \cdots (v_{j+j'},0) $ be the minimal prefix of $(v_j,0) \rho$ such that $v_{j+j'} \in S_{{r(v_j)}-1}\cup \bndr(S_{{r(v_j)}-1})$. 
If $v_{j+j'} \in  S_{{r(v_j)}-1}$ (which implies $j' >0$ due to $v_j \notin S_{r(v_j)-1}$) then we extend $w$ to $w(v_{j+1},0) \cdots (v_{j+j'},0)$ to obtain the next prefix in our inductive construction.
If $v_{j+j'} \in  \bndr(S_{{r(v_j)}-1})$, then there is a vertex~$v_{j+j'+1} \in S_{{r(v_j)}-1}$ and $(v_{j+j'},v_{j+j'+1}) \in D$ due to the definition of the $D$-boundary.
Thus, we extend $w$ to $w(v_{j+1},0) \cdots (v_{j+j'},0)(v_{j+j'+1},1)$ to obtain the next prefix in our inductive construction.
The resulting prefix is  consistent with $\sigma$ and its last vertex is in $S_{{r(v_j)}-1} \subseteq S$, i.e., our invariant is satisfied.
\end{enumerate}

Now, let $v_{j_0}, v_{j_1},v_{j_2},\ldots$ be the sequence of last vertices of the prefixes obtained during the construction.
In particular, $v_{j_0} = v$.
By construction, we have $r(v_{j_0}) > r(v_{j_1}) > r(v_{j_2}) \cdots$.
Hence, we apply the second case at most $r(v_{j_0})$ many times and then have to apply the first case.
Hence, we indeed obtain an infinite play~$\rho$ starting in $v$, which is consistent with $\sigma$ due to our invariant, and which visits $F$, as the first case is eventually applied.
Finally, $\rho$ has at most $r(v_{j_0}) = r(v)$ many disturbances, as each application of the second case adds at most one disturbance edge and the first case adds none.
Thus, $\rho$ witnesses that $\sigma$ is not $(r(v)+1)$-resilient from $v$. 
As we have picked $\sigma$ arbitrarily, we conclude $r_\game(v) \le r(v)$ as desired.

It remains to show  $r_\game(v) = \omega +1 $ for every $v \notin S$.  
We start by listing some properties of such vertices:
\begin{enumerate}

\item\label{propF} $v \notin F$, as $F \subseteq \att(F) = S_0 \subseteq S$.

\item\label{prop0} If $v \in V_0$, then there is a $v'$ with $(v,v') \in E$ and $v' \notin S$. 
Towards a contradiction, assume there is no such $v'$. 
Then, all successors of $v$ are in $S$.
As $v$ has only finitely many successors by assumption on $\arena$, there is a $j$ such that all these successors are in $S_j$.
Hence, $v \in \att(S_j) \subseteq  S_{j+1}\subseteq S$, which contradicts $v \notin S$.

\item\label{prop1} If $v \in V_1$, then all $v'$ with $(v,v') \in E$ satisfy $v' \notin S$. 
Towards a contradiction, assume there is a successor of $v'$ in $S$.
Then, $v'$ is in some $S_j$ and $v \in \att(S_j) \subseteq  S_{j+1}\subseteq S$, which contradicts $v \notin S$.

\item\label{propD} If $v \in V_0$ and $(v,v') \in D$, then $v'  \notin S$.
Again, towards a contradiction  assume there is a disturbance edge leading from $v$ to $v'$ in $S$.
Then,  $v' $ is in some $S_j$ and $v \in \bndr(S_j) \subseteq  S_{j+1}\subseteq S$, which contradicts $v \notin S$.

\end{enumerate}

Thus, due to Property~\ref{prop0}, Player~$0$ must have a positional strategy~$\sigma$ that moves from any vertex~$v \notin S$ to some successor~$v' \notin S$. 
Now, consider a play~$\rho$ that starts in a vertex~$v \notin S$, is consistent with $\sigma$, and has an arbitrary number of disturbances.
It starts outside of $S$, Player~$0$ does not move into $S$ by definition of $\sigma$, Player~$1$ cannot due to Property~\ref{prop1}, and disturbances do not lead into $S$ due to Property~\ref{propD}. Hence, $\rho$ never visits $S$ and thus also avoids $F$, due to Property~\ref{propF}. 
Hence, $\rho$ is winning for Player~$0$.
As $v$ and $\rho$ are arbitrary, we have shown $r_\game(v) = \omega+1$ for every $v \notin S$.
\end{proof}

Finally, the main result of this section shows that optimally resilient strategies exist in all finitely branching safety games, i.e., in particular in pushdown safety games.

\begin{theorem}
\label{theorem_optimalstrategiesinfinitelybranchinggames}
Player~$0$ has positional optimally resilient strategies in finitely branching safety games.
\end{theorem}

\begin{proof}
Let $\game = (\arena, \safety(F))$ with finitely branching $\arena = ( V, V_0, V_1, E, D )$, and let the values~$r(v)$ and the set~$S$ be defined as on Page~\pageref{page_r}.
We have shown $r_\game(v) \le r(v)$ for every $v \in S$ and $r_\game(v) = \omega +1 $ for every $v \notin S$ in the proof of Lemma~\ref{lemma_resiliencevaluessafety}.\ref{lemma_resiliencevaluessafety_finbranch}.
We now show $r_\game(v) \ge r(v)$ for every $v \in S$. 

To simplify our notation, let $X_0 = F$ and $X_{j+1} = S_j \cup \bndr(S_j)$, i.e., $S_j = \att(X_j)$ for every $j$. 
Now, for every $j \in \omega$, let $\sigma_j$ be the trap strategy for Player~$0$ associated with $S_j = \att(X_j)$, i.e., every disturbance-free play that starts in $V \setminus S_j$ and is consistent with $\sigma_j$ never visits $X_j$.
Recall that we defined~$r(v) = \min\set{j \mid v \in S_j}$ for all $v \in S$. 
Thus, if $r(v) > 0$, then $v \notin S_{j-1}$.

We define a positional strategy~$\sigma$ for Player~$0$ as follows: 
\begin{itemize}

\item If $v \in V_0 \cap S$ with $r(v) > 0$ then $\sigma(v) = \sigma_{r(v) - 1}(v)$.

\item If $v \in V_0 \cap S$ with $r(v) = 0$ then $\sigma(v) = v'$ for some arbitrary successor~$v'$ of $v$.

\item If $v \in V_0 \setminus S$ then $\sigma(v) = v'$ for some successor~$v'$ of $v$ with $v' \notin S$. We have argued 
    in the proof of Lemma~\ref{lemma_resiliencevaluessafety}.\ref{lemma_resiliencevaluessafety_finbranch}, that such a successor always exists if $v \notin S$.

\end{itemize}
Fix some $v \in S$ and consider a play~$\rho = (\rho_0,b_0) (\rho_1,b_1) (\rho_2,b_2) \cdots $ starting in $v \in S$, consistent with $\sigma$, and with $k < r(v)$ disturbances. 
A straightforward induction on $j\ge 0$ shows that $r(\rho_j) \ge r(v)- \disturbances((\rho_0,b_0) \cdots  (\rho_j,b_j) )$ for every $j$. 
Thus, $r(\rho_j) \ge r(v) - k > 0 $, which implies $\rho_j \notin F \subseteq S_0$, i.e., $\rho$ is winning for Player~$0$.

Therefore, $\sigma$ is $r(v)$-resilient from every $v \in S$. 
Conversely, in the proof of Lemma~\ref{lemma_resiliencevaluessafety}.\ref{lemma_resiliencevaluessafety_finbranch}, we have shown $r_\game(v) \le r(v)$. 
Hence, $r(v) = r_\game(v)$, i.e., $\sigma$ is $r_\game(v)$-resilient from every $v \in S$.
Furthermore, the arguments presented in the proof of Lemma~\ref{lemma_resiliencevaluessafety}.\ref{lemma_resiliencevaluessafety_finbranch} for vertices~$v \notin S$ show that $\sigma$ is $(\omega+1)$-resilient from every $v \notin S$.

Altogether, $\sigma$ is optimally resilient.
\end{proof}

\section{Characterizing Resilience Values via Classical Games}
\label{sec_riggedgames}
In this section, we characterize the existence of $\alpha$-resilient strategies by games without disturbances.
This generalizes a characterization for $\alpha = \omega+1$ in finite arenas~\cite{NeiderW018} to infinite arenas and all $\alpha \in \omega+2$.

The main idea is to give Player~$1$ control over the disturbances and to restrict the number of their occurrences using the winning condition. 
Intuitively, when it is Player~$0$'s turn at a vertex~$v$, we let Player~$1$ first decide whether to simulate a disturbance edge from $D$ or whether to allow Player~$0$ to pick a standard edge from $E$. 
To this end, we add $v$ to Player~$1$'s vertices and he can either move to some vertex~$v'$ such that the disturbance edge~$(v,v') $ exists.
By doing his, he has to visit the fresh vertex~$(v,v')$, which allows to keep track of the number of simulated disturbances. 
This vertex has exactly one outgoing edge leading to $v'$.
On the other hand, if he does not simulate a disturbance edge, he moves from $v$ to a fresh copy~$\overline{v}$ of $v$ from which Player~$0$ has edges leading to the successors of $v$.
Finally, the moves at Player~$1$'s original vertices are unchanged, but we subdivide the edge so that a play in the extended arena always alternates between vertices from $V$ and auxiliary vertices.

Formally,  given an arena~$\arena = (V, V_0, V_1, E, D)$, we define the rigged arena~$\arena_\rig = (V', V_0', V_1', E', D')$  with $V' = V \cup A$ for the set
\[
A = \set{\overline{v} \mid v \in V_0} \cup D \cup \set{(v,v') \in E \mid v \in V_1}
\]
of auxiliary vertices, $V_0' = \set{\overline{v} \mid v \in V_0}$, $V_1' = V' \setminus V_0'$, $D' = \emptyset$, and $E$ is the union of the following sets of edges:
\begin{itemize}

\item $\set{ (v,(v,v')), ((v,v'),v')  \mid (v,v') \in D}$: Player~$1$ simulates a disturbance edge~$(v,v') \in D$ by moving from $v$ to $v'$ via the auxiliary vertex~$(v,v')$ that signifies that a disturbance is simulated.

\item $\set{ (v,\overline{v}) \mid v\in V_0 }$: Player~$1$ does not simulate a disturbance edge and instead gives control to Player~$0$ by moving to the auxiliary vertex~$\overline{v}$.

\item $\set{  (\overline{v}, v') \mid v \in V_0 \text{ and }  (v,v') \in E }$: Player~$0$ has control at the auxiliary vertex~$\overline{v}$ and simulates a  standard move from $v \in V_0$ to $v'$.

\item $\set{ (v,(v,v')), ((v,v'),v')  \mid (v,v') \in E \text{ and } v\in V_1}$: Player~$1$ simulates a standard move from $v \in V_1$ to $v'$ by moving via the auxiliary vertex~$(v,v')$.

\end{itemize}
We illustrate the definition of the construction of the rigged arena in Figure~\ref{fig_riggedexample}.

\begin{figure*}
\centering
\begin{tikzpicture}

\node at (-1.3, .75) {$q_\initmark$};
\node at (-1.3, -.5) {$\overline{q_\initmark}$};
\node at (-1.3, -2.5) {$q_1$};
\node at (-1.3, -1.25) {$\overline{q_1}$};
\node at (-1.3, -4.5) {$q_2$};
\node at (-1.3, -5.5) {$\overline{q_2}$};

\foreach \i in {0,1,...,7}{
\node[p1s] (ione\i) at (\i*1.5,.75) {};
\node[p0s] (izero\i) at (\i*1.5,-.5) {};
\node[p1s] (1one\i) at (\i*1.5,-2.5) {};
\node[p0s] (1zero\i) at (\i*1.5,-1.25) {};
\node[p1s,fill=black!20] (d\i) at (\i*1.5, -3.75) {};
}

\node 	  (ionedots) at (12,.75) {$\cdots$};
\node 	  (izerodots) at (12,-.5) {$\cdots$};
\node 	  (1onedots) at (12,-2.5) {$\cdots$};
\node 	  (1zerodots) at (12,-1.25) {$\cdots$};
\node 	  (ddots) at (12,-3.75) {$\cdots$};

\node[p1s] (2one0) at (0,-4.5) {};
\node[p0s] (2zero0) at (0,-5.5) {};

\path[-stealth]
(1zero0.west) edge[bend right] (2one0.west)
(2one0) edge[bend left] (2zero0)
(2zero0) edge[bend left] (2one0)
(izero7) edge (ionedots)
(1onedots.west) edge[] (d7.east)
;

\foreach \i in {0,...,7}{
\path[-stealth]
(ione\i) edge (izero\i)
(izero\i.east) edge[bend left] (1one\i.east)
(d\i) edge[] (1one\i)
;
}
\foreach \i in {1,...,7}{
\path[-stealth]
(1one\i) edge[bend left] (1zero\i)
(1zero\i) edge[bend left] (1one\i)
;
}
\path[-stealth]
(1one0) edge (1zero0)
;

\foreach \i [remember=\i as \lastx (initially 0)]in {1,...,7}{
\path[-stealth]
(izero\lastx) edge (ione\i)
;
}

\foreach \i [remember=\i as \lastx (initially 7)]in {6,...,0}{
\path[-stealth]
(1one\lastx.west) edge[] (d\i.east)
;
}

\end{tikzpicture}
\caption{The rigged arena~$\arena_\rig$ for the arena~$\arena$ presented in Figure~\ref{fig_pdsexample}, restricted to vertices reachable from the initial vertex~$(q_\initmark,\bot)$. Round vertices are in $V_0$, square ones in $V_1$, and a gray vertex indicates that a disturbance has been simulated.}
\label{fig_riggedexample}
\end{figure*}
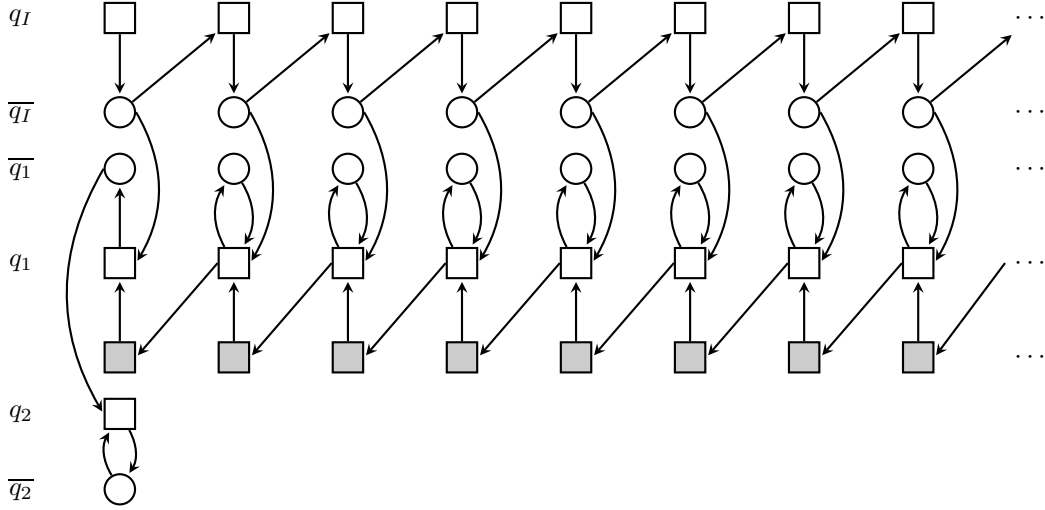

Let $R_{\ge k}$ denote the set of sequences~$v_0v_1v_2 \cdots \in (V')^\omega$ such that $\size{\set{j \mid v_j \in D}} \ge k$, i.e., those plays in which Player~$1$ simulates at least $k$ disturbances.
 Finally, given a winning condition~$\wincond \subseteq V^\omega$ for $\arena$, we define the rigged winning condition
 \[\wincond_\rig = \set{v_0 v_1 v_2 \cdots \in (V')^\omega \mid v_0 \in V \text{ and } v_0 v_2 v_4 \cdots \in \wincond},\]
  which contains all plays in $\arena_\rig$ that start in $V$ and are in $\wincond$ after removing the auxiliary vertices.
Note that $\buechi(D)$ contains those plays that simulate infinitely many disturbances. 

\begin{lemma}
\label{lemma_riggedgamescharacterizeresilience}
Let $\game = (\arena, \wincond)$ be a game, let $v$ be a vertex of $\game$, and let $k \in \omega$.
\begin{enumerate}

\item\label{lemma_riggedgamescharacterizeresilience_omegaplusone} 
Player~$0$ has an $(\omega+1)$-resilient strategy for $\game$ from $v$ if and only if $v \in \winreg_0(\arena_\rig,\wincond_\rig)$.

\item\label{lemma_riggedgamescharacterizeresilience_omega} 
Player~$0$ has an $\omega$-resilient strategy for $\game$ from $v$ if and only if  $v \in \winreg_0(\arena_\rig,\wincond_\rig \cup \buechi(D))$.

\item\label{lemma_riggedgamescharacterizeresilience_k} 
Player~$0$ has a $k$-resilient strategy for $\game$ from $v$ if and only if $v \in \winreg_0(\arena_\rig, \wincond_\rig \cup R_{\ge k})$.
\end{enumerate}
\end{lemma}

\begin{proof}
We begin by introducing translations between plays that are useful in all three cases.

First, we translate a play prefix~$w$ in $\arena$ into a play prefix~$t'(w)$ in $\arena_\rig$ satisfying the following invariant: 
\[t'( (v_0, b_0) \cdots (v_j, b_j) )\]
 starts in $v_0$ and ends in $v_j$.
We proceed by induction starting with $t'(v_0, b_0) = (v_0, 0)$.
For the induction step, we have to consider a play prefix~$(v_0, b_0) \cdots (v_j, b_j) (v_{j+1}, b_{j+1})$ such that $t'((v_0, b_0) \cdots (v_j, b_j))$ is already defined, which ends in $v_j$ due to our invariant.
We consider several cases:
\begin{itemize}

\item If $b_{j+1} = 1$, then $(v_j, v_{j+1})$ is a disturbance edge, which is simulated in $\arena_\rig$ by Player~$1$ taking control at $v_j$, moving to $(v_j, v_{j+1})$ and then to $v_{j+1}$. 
Hence, we define 
\[
t'((v_0, b_0) \cdots (v_j, b_j) (v_{j+1}, b_{j+1})) = t'((v_0, b_0) \cdots (v_j, b_j) ) \cdot ((v_j,v_{j+1}),0) (v_{j+1},0). 
\]

\item If $b_{j+1} = 0$ and $v_j \in V_0$, then $(v_j, v_{j+1})$ is a non-distur\-bance edge picked by Player~$0$, which is simulated in $\arena_\rig$ by Player~$1$ ceding control at $v_j$ to Player~$0$ by moving $\overline{v_j}$, from where Player~$0$ can then move to $v_{j+1}$.
Hence, we define
\[
t'((v_0, b_0) \cdots (v_j, b_j) (v_{j+1}, b_{j+1})) = t'((v_0, b_0) \cdots (v_j, b_j) ) \cdot (\overline{v_j},0) (v_{j+1},0). 
\]

\item If $b_{j+1} = 0$ and $v_j \in V_1$, then $(v_j, v_{j+1})$ is a non-distur\-bance edge picked by Player~$1$, which is simulated in $\arena_\rig$ by Player~$1$ directly moving to $v_{j+1}$.
Hence, we define
\[
t'((v_0, b_0) \cdots (v_j, b_j) (v_{j+1}, b_{j+1})) = t'((v_0, b_0) \cdots (v_j, b_j) ) \cdot ((v_j, v_{j+1}),0)  (v_{j+1},0). 
\]
\end{itemize}
In each case, the invariant is satisfied and 
\[t'((v_0, b_0) \cdots (v_j, b_j) (v_{j+1}, b_{j+1}))\] is indeed a play prefix due to $t'((v_0, b_0) \cdots (v_j, b_j))$ ending in $v_j$.

Furthermore, we extend $t'$ to infinite plays by defining $t'((v_0,b_0)(v_1,b_1)(v_2,b_2) \cdots)$ to be the unique play~$\rho'$ in $\arena_\rig$ such that $t'((v_0,b_0)\cdots(v_j,b_j))$ is a prefix of $\rho'$ for every $j \in \omega$.
Let $\rho = (v_0,0) (v_1,0) (v_2,0) \cdots$ be a play in $\arena$.
Then we have $t'(\rho) = (v_0, 0) (a_0,0) (v_1, 0) (a_1,0) (v_2, 0) (a_2,0)\cdots$ for  auxiliary vertices~$a_0a_1a_2\cdots$ and $\disturbances(\rho) = \size{\set{ j \mid a_j \in D }}$, i.e., the number of disturbances during a play~$\rho$ in $\arena$ is equal to the number of vertices from $D \subseteq A$ occurring in $t'(\rho)$.

Finally, we can use the translation~$t'$ to transform a strategy~$\sigma'$ for Player~$0$ in $\arena_\rig$ to a strategy~$\sigma$ for her in $\arena$. 
To this end, let $b^-$ denote the homomorphism from $(V' \times \set{0,1})^*$ to $(V')^*$ that removes the second component.
Then, we define
\[
\sigma(v_0 \cdots v_j) = \sigma'( b^-( t'( (v_0, b_0) \cdots (v_j, b_j) ) ) \cdot \overline{v_j})
\]
where $b_0 = 0$ and for every $0< j' \le j$, $b_{j'} = 1 $ if and only if $v_{j'-1} \in V_0$ and $v_{j'} \neq \sigma(v_0 \cdots v_{j'-1}) $, i.e., we reconstruct the consequential disturbances with respect to $\sigma$ as defined thus far. 
A simple induction shows that a play~$\rho$ in $\arena$ being consistent with $\sigma$ implies that $t'(\rho)$ in $\arena_\rig$ is consistent with $\sigma'$.

Now, we consider the other direction and translate a play prefix~$w$ in $\arena_\rig$ into a play prefix~$t(w)$ in $\arena$. 
Here, we only consider play prefixes~$w$ starting and ending in a vertex from $V' \setminus A$, i.e., only play prefixes
that do not start or end in one of the auxiliary vertices.
This satisfies the following invariant: $t( (v_0, 0) \cdots (v_j, 0) )$ starts in $v_0$ and ends in $v_j$ (recall that $\arena_\rig$ has no disturbance edges, which implies that all bits~$b_j$ in $w$ are equal to zero).
Again, we proceed by induction and start with $t(v_0,0) = (v_0,0)$.
For the induction step, consider a play prefix~$(v_0, 0) \cdots (v_j, 0)(a_j,0) (v_{j+1}, 0)$ such that $t((v_0, b_0) \cdots (v_j, b_j))$ is already defined, which ends in $v_j$ due to our invariant.

\begin{itemize}
	
	\item If the prefix is of the form 
	\[(v_0,0) \cdots (v_j,0)((v_j,v_{j+1}),0)(v_{j+1},0) \]
	 with $v_j \in V_0$, then the last move simulated during the play prefix is the disturbance edge~$(v_j,v_{j+1}) \in D$.
Hence, we define 
\[
t((v_0,0) \cdots (v_j,0)((v_j,v_{j+1}),0)(v_{j+1},0) ) =t((v_0,0) \cdots (v_j,0) ) \cdot (v_{j+1},1).
\]
	
	\item If the prefix is of the form
	\[(v_0,0) \cdots (v_j,0)(\overline{v_j},0)(v_{j+1},0),\] 
	then the last move simulated during the play prefix is the non-disturbance edge~$(v_j,v_{j+1}) \in E$ with $v_j \in V_0$.
Hence, we define 
\[
t((v_0,0) \cdots (v_j,0)(\overline{v_j},0)(v_{j+1},0) ) = t((v_0,0) \cdots (v_j,0) ) \cdot (v_{j+1},0).
\]	

	\item If the prefix is of the form 
	\[(v_0,0) \cdots (v_j,0)((v_j,v_{j+1}),0)(v_{j+1},0) \] with $v_j \in V_1$, then the last move simulated during the play prefix is the non-disturbance edge~$(v_j,v_{j+1}) \in E$.
Hence, we define 
\[
t((v_0,0) \cdots (v_j,0)(v_j,v_{j+1})(v_{j+1},0) ) = t((v_0,0) \cdots (v_j,0) ) \cdot (v_{j+1},0).
\]
	
\end{itemize}
In each case, the invariant is satisfied and the extension is indeed a play prefix due to $t((v_0,0) \cdots (v_j,0) ) $ ending in $v_j$.

Again, we extend the function~$t$ to infinite plays by defining $t((v_0,0)(v_1,0)(v_2,0) \cdots)$ to be the unique play~$\rho$ in $\arena$ such that $t((v_0,0)\cdots(v_j,0))$ is a prefix of $\rho$ for every $j \in \omega$.
Let \[\rho' = (v_0, 0) (a_0, 0) (v_1, 0) (a_1,0) (v_2, 0) (a_2,0) \cdots\] be a play in $\arena_\rig$ starting in $V$.
Hence, $t(\rho') = (v_0,b_0) (v_1, b_1) (v_2, b_2) \cdots$ for some bits~$b_j$, and $\size{\set{ j \mid a_j \in D }} = \disturbances(t(\rho')) $, i.e., the number of vertices from $D \subseteq A$ occurring in $\rho'$ is equal to the number of disturbances during the play~$t(\rho')$ in $\arena$.

To conclude, we again show that we can use the translation~$t$ to transform a strategy~$\sigma$ for Player~$0$ in $\arena$ to a strategy~$\sigma'$ for her in $\arena$. 
Here, let $b^-$ denote the homomorphism from $(V \times \set{0,1})^*$ to $V^*$ that removes the second component in each letter.
Now, we define
\[
\sigma'(v_0 \cdots v_j\overline{v_j}) = \sigma( b^-( t( (v_0, 0) \cdots (v_j, 0) ) )).
\]
Finally, a simple induction shows that a play~$\rho'$ in $\arena_\rig$ being consistent with $\sigma'$ implies that $t(\rho')$ in $\arena$ is consistent with $\sigma$.

After these preparations, the proof of the three characterizations is straightforward employing the transformation of strategies described above.

\ref{lemma_riggedgamescharacterizeresilience_omegaplusone}.)
Let $v \in \winreg_0(\arena_\rig,\wincond_\rig)$, i.e., Player~$0$ has a winning strategy~$\sigma'$ from $v$. 
Let the strategy~$\sigma$ for Player~$0$ in $\arena$ be obtained from $\sigma'$ as described above. 
We claim that $\sigma$ is $(\omega+1)$-resilient from $v$. 
To this end, let $\rho = (v_0, b_0) (v_1, b_1) (v_2, b_2)\cdots$ be a play in $\game$ that starts in $v$, is consistent with $\sigma$, and has an arbitrary number of disturbances. 
We need to show that $\rho$ is winning for Player~$0$, i.e., $v_0 v_1 v_2 \cdots \in \wincond$.

As argued above, the play~$t'(\rho)$ in $\arena_\rig$ is of the form
\[(v_0, 0) (a_0, 0)(v_1, 0) (a_1, 0)(v_2, 0) (a_2, 0) \cdots,\] starts in $v$, and is consistent with $\sigma'$.
This implies $t'(\rho) \in \wincond_\rig$. 
So, by definition of $\wincond_\rig$, we have indeed $v_0v_1v_2\cdots \in \wincond$.

Now, assume Player~$0$ has an $(\omega+1)$-resilient strategy~$\sigma$ for $\game$ from $v$.
Let the strategy~$\sigma'$ for Player~$0$ in $\arena_\rig$ be obtained from $\sigma$ as described above.
We claim that $\sigma'$ is a winning strategy from $v$ in the game~$(\arena_\rig,\wincond_\rig)$. 
To this end, let $\rho' = (v_0,0) (a_0,0) (v_1,0) (a_1,0) (v_2,0) (a_2,0)\cdots$ be a play in $\arena_\rig$ starting in $v$ and consistent with $\sigma'$. 
We need to show that $\rho'$ is winning for Player~$0$.

As argued above, the play~$t(\rho') = (v_0,b_0) (v_1,b_1) (v_2,b_2)\cdots $ in $\arena$ starts in $v$ and is consistent with $\sigma$.
Since $\sigma$ is $(\omega+1)$-resilient from $v$, $t(\rho')$ is winning for Player~$0$, as it has at most $\omega$ disturbances.
Thus, $v_0 v_1 v_2 \cdots \in \wincond$.
Hence, $\rho' \in \wincond_\rig$ by definition of $\wincond_\rig$, i.e., $\rho'$ is indeed winning for Player~$0$.

\ref{lemma_riggedgamescharacterizeresilience_omega}.)
As this proof is a refinement of the previous one, we only sketch the differences.

First, let $v \in \winreg_0(\arena_\rig,\wincond_\rig \cup \buechi(D))$, i.e., Player~$0$ has a winning strategy~$\sigma'$ from $v$ which induces a strategy~$\sigma$ for her in $\arena$.
We show that $\sigma$ is $\omega$-resilient from $v$.
To this end, let $\rho = (v_0, b_0) (v_1, b_1) (v_2, b_2)\cdots$ be a play in $\game$ that starts in $v$, is consistent with $\sigma$, and has a finite number of disturbances. 
We need to show that $\rho$ is winning for Player~$0$.

Again, the play~$t'(\rho)$ in $\arena_\rig$ starts in $v$ and is consistent with $\sigma'$.
Now, we additionally have that $t'(\rho)$ visits vertices in $D$ only finitely often, as the number of these visits is equal to the number of disturbances in $\rho$, as argued above.
Hence, $t'(\rho)$ is not in $\buechi(D)$, which implies $t'(\rho) \in \wincond_\rig$, as $t'(\rho)$ is consistent with the winning strategy~$\sigma$. 
This allows us, as before, to conclude that $\rho$ is indeed winning for Player~$0$.

Now, assume Player~$0$ has an $\omega$-resilient strategy~$\sigma$ for $\game$ from $v$ and let $\sigma'$ be the induced strategy for her in $\arena_\rig$.
We show that $\sigma'$ is winning from $v$ in the game~$(\arena_\rig,\wincond_\rig)$, i.e., every play $\rho' = (v_0,0)(a_0,0) (v_1,0)(a_1,0) (v_2,0)(a_2,0)\cdots$ in $\arena_\rig$ starting in $v$ and consistent with $\sigma'$ is winning for Player~$0$.

If $v_0 a_0 v_1 a_1 v_2 a_2 \cdots$ is in $\buechi(D)$, then $\rho'$ is winning for Player~$0$.
Thus, assume it is not. 
Then, consider the play~$t(\rho') = (v_0,b_0) (v_1,b_1) (v_2,b_2)\cdots $ in $\arena$.
It starts in $v$, is consistent with $\sigma$, and has the same finite number of disturbances as $\rho'$ has visits to vertices in $D$.
Hence, as $\sigma$ is $\omega$-resilient from $v$, $t(\rho')$ is winning for Player~$0$.
From this we can conclude, as before, that $\rho'$ is indeed winning for Player~$0$.

\ref{lemma_riggedgamescharacterizeresilience_k}.)
 Analogously to the previous one arguing about \myquot{less than $k$ disturbances} instead of \myquot{finitely many disturbances}.
\end{proof}

\section{Resilience in Pushdown Safety Games}
\label{sec_pushdown}
The goal of this section is to develop an algorithm that determines the resilience of the initial vertex of a pushdown safety game. 
To this end, we rely on the characterizations presented in the previous section as as well as an upper bound on the possible finite resilience values that can be realized by the initial vertex of such a game.
We begin by showing that the first two characterizations presented in Lemma~\ref{lemma_riggedgamescharacterizeresilience} (for $\omega+1$ and $\omega$) are effective for pushdown games.
Intuitively, we prove that a pushdown machine~$\pds$ inducing an arena~$\arena$ can in polynomial time be turned into a pushdown machine~$\pds_\rig$ inducing the arena~$\arena_\rig$. 

We state the result for parity conditions, which subsume safety conditions.

\begin{lemma}
\label{lemma_riggedpushdowngameseffective}
The following problem is $\exptime$-complete (and $\pspace$-complete if inputs are restricted to one-counter games):
\myquot{Given a pushdown parity game~$\game$ with initial vertex~$v_\initmark$ and $\alpha \in \set{\omega, \omega+1}$, does Player~$0$ have an $\alpha$-resilient strategy for $\game$ from $v_\initmark$?}.
If yes, such a strategy can be computed in exponential time.
\end{lemma}

\begin{proof}
Given a pushdown arena~$\arena$ induced by a PDS~$\pds$ with set~$Q$ of states, a partition~$\set{Q_0, Q_1}$ of $Q$, and a transition relation~$\Delta$ inducing the disturbance edges, a PDS~$\pds'$ with set~$Q'$ of states and a partition~$\set{Q_0', Q_1'}$ of $Q'$ inducing~$\arena_\rig$ can be computed in linear time. 
If $\pds$ is one-counter, then so is $\pds'$.
Further, given a coloring~$\col$ of $Q$, one can determine 
\begin{itemize}

\item a coloring~$\col'$ of $Q'$ such that $\parity(\col') = \parity(\col)_\rig $, and

\item a coloring~$\col''$ of $Q'$ such that $\parity(\col'') = \parity(\col)_\rig\cup \buechi(D)$, where $D$ is the set of disturbances edges of $\arena$.

\end{itemize}
In $\col'$, all vertices in $V$ inherit their colors from $\col$ and auxiliary vertices are colored by zero, which makes them irrelevant, while in $\col''$, all vertices in $V$ inherit their colors from $\col$, all vertices in $D$ are assigned an even color that is larger than all colors in $\col$'s range, and all other auxiliary vertices are colored by zero.

Hence, the games characterizing the existence of $(\omega+1)$-resilient and $\omega$-resilient strategies are pushdown (one-counter) parity games that can be efficiently constructed.
Finally, checking whether Player~$0$ wins a pushdown parity game from the initial vertex is $\exptime$-complete~\cite{Walukiewicz01} while checking whether Player~$0$ wins a one-counter parity game from the initial vertex is $\pspace$-complete~\cite{JancarS07,Serre06}.
Furthermore, the first algorithm directly yields winning strategies for the rigged games, which can easily be turned into $(\omega+1)$-resilient or $\omega$-resilient strategies for the original game.

The lower bounds hold already
for determining the winner
of a disturbance-free pushdown (one-counter) safety game,
which is hard for $\exptime$~\cite{Walukiewicz01} ($\pspace$~\cite{JancarS07}~\footnote{The result cited pertains to emptiness of alternating word automata over a singleton alphabet. However it is easy to see that this problem can be reduced to solving one-counter safety games.}).
\end{proof}

Both $\exptime$-hardness and $\pspace$-hardness already hold for pushdown safety games and one-counter safety games, respectively.
The third characterization of Lemma~\ref{lemma_riggedgamescharacterizeresilience} (for $k \in \omega$) is effective as well (even for parity games).
Here the running time depends on $k$.

\begin{lemma}
\label{lemma_riggedpushdowngameseffective_k}
The following problem is in $\twoexp$ (in $\expspace$ if the input is one-counter): 
\myquot{Given a pushdown parity game~$\game$ with initial vertex~$v_\initmark$ and $k \in \omega$ (encoded in binary), does Player~$0$ have a $k$-resilient strategy for $\game$ from $v_\initmark$?}.
If yes, such a strategy can be computed in doubly-exponential time.
\end{lemma}

\begin{proof}
Assume the input~$\game = (\arena, \parity(\col))$ is induced by a PDS~$\pds$ with set~$Q$ of states, a partition~$\set{Q_0, Q_1}$ of $Q$, and a coloring~$\col$ of $Q$.
Then, we construct a PDS~$\pds'$ with set~$Q'$ of states and a partition $\set{Q_0', Q_1'}$ of $Q'$ inducing~$\arena_\rig$ as for the proof of Lemma~\ref{lemma_riggedgamescharacterizeresilience}.
Now, we turn $\pds'$ into a PDS~$\pds'_k$ with set~$Q' \times \set{0,\ldots, k}$ of states which uses the additional component to keep track of the number of simulated disturbances,  up to $k$.
Further, we use the partition 
\[\set{Q_0' \times \set{0,\ldots, k}, Q_1' \times \set{0,\ldots, k} }\] and define the coloring~$\col'$ such that $\col'(q,k') = \col(q)$ for $k' < k$ and $\col'(q,k) = 1$.

The resulting pushdown game is equivalent to $(\arena_\rig, \wincond_\rig \cup R_{\ge k})$ and the winner from the initial vertex~$((q_\initmark, 0), \bot)$ can be determined in exponential time in $k$ and the size of $\pds$~\cite{Walukiewicz01}, i.e., in doubly-exponential time in the size of the input, as $k$ is encoded in binary.
Due to Lemma~\ref{lemma_riggedgamescharacterizeresilience}.\ref{lemma_riggedgamescharacterizeresilience_k}, Player~$0$ wins from the initial vertex if and only if she has a $k$-resilient strategy from $v_\initmark$ in $\game$, i.e., if and only if $r_\game(v_\initmark) \ge k$. 
Furthermore, the algorithm computes winning strategies for Player~$0$ in doubly-exponential time, if they exist at all.
These can easily be turned into $k$-resilient strategies for the original game.

If the input is one-counter, then the resulting pushdown game is one-counter as well and the winner from the initial vertex can be determined in polynomial space in $k$ and the size of $\pds$~\cite{Serre06}, i.e., in exponential space in the input.
\end{proof}

There are no vertices of resilience $\omega$ in pushdown safety games (Lemma~\ref{lemma_resiliencevaluessafety}.\ref{lemma_resiliencevaluessafety_finbranch}). 
Thus, the effective characterizations we have presented so far suffice to determine the resilience of the initial vertex in such a game:
First, check whether it is $\omega+1$; if not, then it has to be finite. 
Hence, for increasing $k$, check whether the resilience is greater than $k$.
As the resilience is finite, this algorithm will eventually terminate and report the resilience correctly. 
However, without an upper bound on the possible finite resilience values of the initial vertex, there is no bound on the running time, just a termination guarantee.
In the remainder of this section, we present a tight doubly-exponential upper bound~$b(\pds)$ on the resilience of the \emph{initial vertex} in pushdown safety games in the case the resilience is finite.
That is, if $r_\game(v_\initmark) \in \omega$ then $r_\game(v_\initmark) < b(\pds)$.
Note that any proof of the upper bound has to depend on the vertex under consideration being initial, as we have shown that there is in general no upper bound on finite resilience values assumed in pushdown safety games (cf.~Example~\ref{example_safetyallpossibleresiliencevalues}).
The bound~$b(\pds)$ only depends on the pushdown system~$\pds$  inducing the game and yields an effective algorithm to determine the resilience of the initial vertex~$v_\initmark$, presented as Algorithm~\ref{algorithm_ocssafety}.

\begin{algorithm}
\begin{algorithmic}[1]
 \IF {$v_\initmark \in \winreg_0(\arena_\rig, \safety(F)_\rig) $}
 \RETURN {$\omega+1$}
 \ENDIF
 \FOR {$k = 1 $ \TO $b(\pds)$}
 	\IF { $v_\initmark\in\winreg_1(\arena_\rig, \safety(F)_\rig \cup R_{\ge k})$ }
 	\RETURN {$k-1$}
 	\ENDIF
 \ENDFOR
\end{algorithmic} 
\caption{Computing the resilience of the initial vertex~$v_\initmark$ of a pushdown safety game~$\game = (\arena, \safety(F))$ induced by a PDS~$\pds$}.
\label{algorithm_ocssafety}
\end{algorithm}

Given a PDS~$\pds$ with set~$Q$ of states and set~$\Gamma$ of stack symbols let $\pds_\rig$ be the PDS obtained from $\pds$ by implementing the transformation from an arena to the rigged arena.
The cardinality of the set~$Q'$ of states of $\pds_\rig$ is bounded quadratically in $\size{Q}$ and the set of stack symbols used by $\pds_\rig$ is still $\Gamma$.
We define $b(\pds) = \size{Q'} \cdot h(\pds)\cdot \size{\Gamma}^{h(\pds)}$, where $h(\pds) = \size{Q'}\cdot\size{\Gamma}\cdot 2^{\size{Q'}+1} +1$\label{bounddef}. 
Note that $b(\pds) \in 2^{2^{\bigo(\size{\pds}^2)}}$ and $b(\pds) \in {2^{\bigo(\size{\pds}^2)}}$ if $\pds$ is an OCS.

\begin{lemma}
\label{lemma_ubresilienceinitialvertexsafety}
Let $\game$ be a pushdown safety game with initial vertex~$v_\initmark$. If $r_\game(v_\initmark) \neq \omega+1$, then $r_\game(v_\initmark) < b(\pds)$, where $\pds$ is the PDS underlying $\game$.
\end{lemma}

To prove this result, we apply a result about winning strategies for Player~$1$ in pushdown safety games (Player~$1$ has a reachability condition in a safety game: he wins if $F$ is visited at least once).
Fix a disturbance-free pushdown safety game~$\game = (\arena, \safety(F))$ with initial vertex~$v_\initmark$.
We say that a winning strategy~$\tau$ for Player~$1$ from $v_\initmark$ \emph{bounds the stack height to} $n \in \omega$ if every play~$v_0 v_1 v_2 \cdots$ that starts in $v_\initmark$ and is consistent with $\tau$ satisfies the following condition for all $j \in \omega$: either there is some $j' \le j$ with $v_{j'} \in F$ or $\sh(v_j) \le n$.
Thus, such a strategy ensures a visit to $F$ when starting in the initial vertex, and ensures that the stack height~$n$ is never exceeded before $F$ is visited for the first time.
The next proposition shows that such a strategy always exists for $n = h(\pds)$, if Player~$1$ wins from $v_\initmark$ at all.

\begin{lemma}
\label{lemma_ubstackheightreachability}
If $v_\initmark \in \winreg_1(\game)$, then Player~$1$ has a winning strategy~$\tau$ that bounds the stack height to $h(\pds)$, where $\pds$ is the PDS underlying $\game$.
\end{lemma}

\begin{proof}
We transform $\game$ into a parity game as described at the end of Section~\ref{subsec_games} on Page~\pageref{cons_safety2parity}.
This transformation can be implemented on the PDS inducing $\game$ without increasing the number of states or the number of stack symbols.
Furthermore, the parity condition only uses two colors, say $0$ for states outside of $F$ and $1$ for states in $F$, which are sinks.
Now, the desired result follows from a result on the existence of strategies in pushdown games that bound the occurrence of undesirable colors (here, the color~$0$, which is undesirable for Player~$1$)~\cite{FridmanZ12}. 
Slightly more formally, in the resulting parity game, the \emph{stair score} for the color~$0$ after a play prefix (see \cite{FridmanZ12} for definitions) is equal to the stack height of the prefix. 
Now, the main result in the work cited above shows that Player~$1$ has a strategy that bounds the stair score  for $0$ by $h(\pds)$, if he wins at all. 
Thus, this strategy bounds the stack height to $h(\pds)$.
\end{proof}

Now, we are able to prove the upper bound~$b(\pds)$ on the resilience of the initial vertex of a pushdown safety game induced by $\pds$ in case this value is finite. 

\begin{proof}[Proof of Lemma~\ref{lemma_ubresilienceinitialvertexsafety}]
Let  $ r_\game(v_\initmark) \neq \omega+1 $.
As pushdown arenas are finitely branching, Lemma~\ref{lemma_resiliencevaluessafety} yields $ r_\game(v_\initmark) \in \omega$, say $r_\game(v_\initmark)  = k$.
By definition, Player~$0$ has a $k$-resilient strategy for $\game$ from $v_\initmark$, but no $(k+1)$-resilient strategy. 
Hence, due to Lemma~\ref{lemma_riggedgamescharacterizeresilience}.\ref{lemma_riggedgamescharacterizeresilience_k}, Player~$1$ wins the  game 
\[(\arena_\rig, \safety(F)_\rig \cup R_{\ge k+1})\] from $v_\initmark$.
Thus, he also wins the safety game
\[(\arena_\rig, \safety(F)_\rig )\] 
from $v_\initmark$, as every winning strategy for Player~$1$ for  the former game is also one for the latter.
Hence, applying Lemma~\ref{lemma_ubstackheightreachability} yields the existence of a winning strategy~$\tau$ for the latter game from $v_\initmark$ that bounds the stack height by $h(\pds)$.
Note that we can assume $\tau$ to be positional (see Lemma~\ref{lemma_positionalstategies} on Page~\pageref{lemma_positionalstategies} for a stronger statement and note that the construction presented in its proof preserves bounds on the stack height).

Now, every play that starts in $v_\initmark$ and is consistent with $\tau$ visits each vertex with stack height at most $h(\pds)$ at most once before reaching $F$.
There are at most $b(\pds)$ such vertices, i.e., after at most $b(\pds)-1$ moves, $F$ is reached.

Now, we show that Player~$0$ has no $b(\pds)$-resilient strategy from $v_\initmark$ in $\game$. 
To this end, we show for that every strategy~$\sigma$ for her, there is a play~$\rho$ that starts in $v_\initmark$, is consistent with $\sigma$, has at most $b(\pds)-1$ many disturbances, and visits a vertex in $F$, i.e., it is losing for Player~$0$.

Let $\sigma'$ be the strategy for Player~$0$ in $\arena_\rig$ obtained by transforming $\sigma$ as described in the proof of Lemma~\ref{lemma_riggedgamescharacterizeresilience}.
Now, let $\rho'$ be the unique play of $\arena_\rig$ starting in $v_\initmark$ that is consistent with $\sigma'$ and $\tau$, which visits $F$ after at most $b(\pds)-1$ many moves. 
Hence, there are at most $b(\pds)-1$ many simulated disturbances in $\rho'$ before the first visit to $F$.
Now, $t(\rho')$ starts in $v$, is consistent with $\sigma$, and there are at most $b(\pds)-1$ many disturbances in $t(\rho')$ before the first visit to $F$ (which occurs).
Now, we just replace the suffix of $t(\rho')$ after the first visit to $F$ by some disturbance-free suffix so that the resulting play~$\rho$ is still consistent with $\sigma$.
We obtain a play~$\rho$ starting in $v_\initmark$, consistent with $\sigma$, with at most $b(\pds)-1$ many disturbances that is losing for Player~$0$.
Hence, $\sigma$ is indeed not $b(\pds)$-resilient.
As we have picked $\sigma$ arbitrarily, there is no $b(\pds)$-resilient strategy from $v_\initmark$ and therefore $r_\game(v_\initmark) < b(\pds)$.
\end{proof}

This upper bound immediately implies correctness of  Algorithm~\ref{algorithm_ocssafety}, which determines the resilience of the initial vertex of a pushdown safety game.

\begin{theorem}
\label{theorem_determiningresiliencesafetypds}
The following problem can be solved in triply-exponential time: 
\myquot{Given a pushdown safety game~$\game$ with initial vertex~$v_\initmark$, determine $r_\game(v_\initmark)$}.
If yes, an $r_\game(v_\initmark)$-resilient strategy can be computed in triply-exponential time.
\end{theorem}

\begin{proof}
Algorithm~\ref{algorithm_ocssafety} is correct due to 
Lemma~\ref{lemma_ubresilienceinitialvertexsafety}.
The triply-exponential running time stems from the doubly-exponential bound $b(\pds)$ presented in Lemma~\ref{lemma_ubresilienceinitialvertexsafety}, which has to be plugged into Lemma~\ref{lemma_riggedpushdowngameseffective_k} to implement the check in Line~$4$.
The check in Line~$1$ runs in exponential time (Lemma~\ref{lemma_riggedpushdowngameseffective}) and the for-loop terminates after at most doubly-exponentially many iterations.
\end{proof}

Note that there is a gap between the triply-exponential upper bound and the exponential lower bound obtained for the related decision problems for $\omega$ and $\omega+1$ (Lemma~\ref{lemma_riggedpushdowngameseffective}).

The complexity for the special case of one-counter safety games is much smaller, i.e., the resilience of the initial vertex can be computed in exponential space, as the winner of one-counter safety games can be computed in polynomial space~\cite{Serre06} and the upper bound on finite resilience values of the initial vertex is only exponential.
Furthermore, a witnessing strategy can be computed in doubly-exponential time using Lemma~\ref{lemma_riggedpushdowngameseffective_k}.
In the next section, we prove that one can do even better by exploiting the simple structure of one-counter arenas.

To conclude this section, we claim that the bound~$b(\pds)$ on the resilience of an initial vertex in a pushdown safety game with finite resilience is tight:
There is an exponential lower bound for the one-counter case and a doubly-exponential lower bound for the pushdown case. 
Both constructions are generalizations of constructions that appeared in the literature previously~\cite{CarayolH18}.
To simplify our notation, let $p_j$ denote the $j$-th prime number and define the primorial~$\prim{k} = \Pi_{j=1}^k p_j $ to be the product of the first $k$ prime numbers. 
We have $\prim{k} \ge 2^k$.

\begin{lemma}
\label{lemma_lbresilienceinitialvertexsafety}
Let $k \in \omega$.
\begin{enumerate}
	
	\item\label{lemma_lbresilienceinitialvertexsafety_ocs}
There is a one-counter safety game $\game_k$ with initial state~$v_\initmark$ such that $r_\game(v_\initmark) = \prim{k}$ and the underlying OCS has polynomially many states in $k$.
	
	\item\label{lemma_lbresilienceinitialvertexsafety_pds}
There is a pushdown safety game $\game_k'$ with initial state~$v_\initmark$ such that $r_\game(v_\initmark) = 2^{\prim{k}} -1$ and the underlying PDS has polynomially many states in $k$ and two stack symbols.
	
\end{enumerate}	
\end{lemma}

\begin{proof}
\ref{lemma_lbresilienceinitialvertexsafety_ocs}.)
We show the game~$\game_2$ in Figure~\ref{fig_safetyresillowerbounds_ocs} and later explain the general case. 
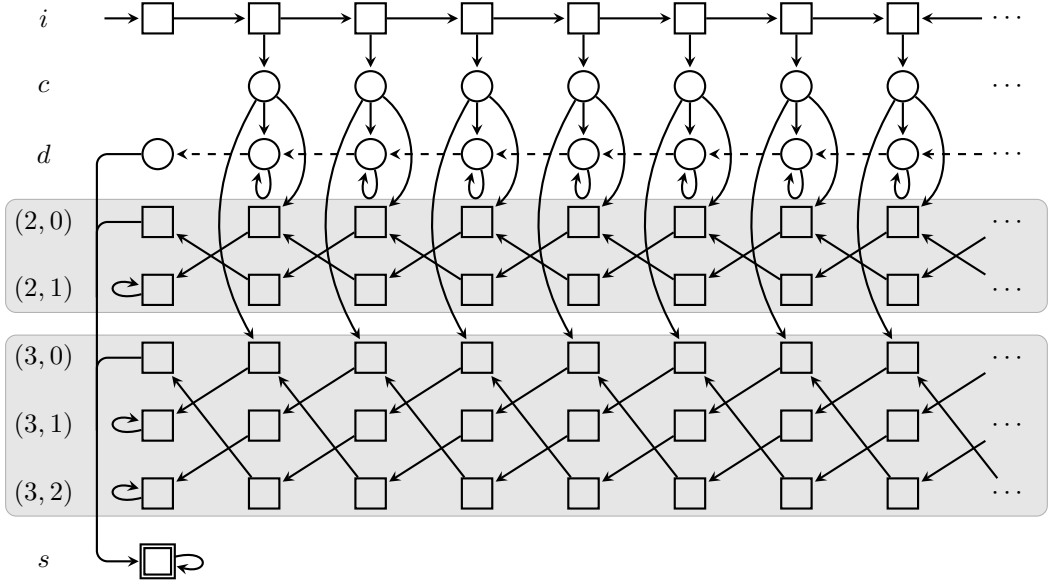
\begin{figure*}
\centering

\def\x{1.4}
\def\y{-.9}

\begin{tikzpicture}[]

\draw[black!30, fill = black!10, rounded corners] (-2,-2.4) rectangle (11.7,-3.9);

\draw[black!30, fill = black!10, rounded corners] (-2,-4.2) rectangle (11.7,-6.6);

\foreach \i in {1,2,...,7}{
	\node[p0s] (c\i) at (\x*\i,1*\y) {};
}

\foreach \i in {0,1,...,7}{
	\node[p1s] (i\i) at (\x*\i,0*\y) {};
	\node[p0s] (d\i) at (\x*\i,2*\y) {};

	\node[p1s] (20\i) at (\x*\i,3*\y) {};
	\node[p1s] (21\i) at (\x*\i,4*\y) {};

	\node[p1s] (30\i) at (\x*\i,5*\y) {};
	\node[p1s] (31\i) at (\x*\i,6*\y) {};
	\node[p1s] (32\i) at (\x*\i,7*\y) {};
}

\node 	 (idots) at (11.2,0*\y) {$\cdots$};
\node 	 (cdots) at (11.2,1*\y) {$\cdots$};
\node	 (ddots) at (11.2,2*\y) {$\cdots$};

\node	 (20dots) at (11.2,3*\y) {$\cdots$};
\node	 (21dots) at (11.2,4*\y) {$\cdots$};

\node	 (30dots) at (11.2,5*\y) {$\cdots$};
\node	 (31dots) at (11.2,6*\y) {$\cdots$};
\node	 (32dots) at (11.2,7*\y) {$\cdots$};

\node[p1s, double]	 (s) at (0,8*\y) {};

	\node (il) at (-1.5,0*\y) {$i$};
	\node (cl) at (-1.5,1*\y) {$c$};
	\node (dl) at (-1.5,2*\y) {$d$};

	\node (20l) at (-1.5,3*\y) {$(2,0)$};
	\node (21l) at (-1.5,4*\y) {$(2,1)$};

	\node (30l) at (-1.5,5*\y) {$(3,0)$};
	\node (31l) at (-1.5,6*\y) {$(3,1)$};
	\node (32l) at (-1.5,7*\y) {$(3,2)$};
	
	\node	 (sl) at (-1.5,8*\y) {$s$};

\foreach \i in {1,...,7}{
\path[-stealth]
(i\i) edge (c\i)
(d\i) edge[loop below] ()
(c\i) edge (d\i)
(c\i) edge[bend left=50] (20\i)
(c\i) edge[bend right] (30\i)

;
}

\foreach \i [remember=\i as \lastx (initially 0)]in {1,...,7}{
\path[-stealth]
(i\lastx) edge (i\i)
;
}

\foreach \i [remember=\i as \lastx (initially 7)]in {6,...,0}{
\path[-stealth]
(d\lastx) edge[fault] (d\i)
(20\lastx) edge (21\i)
(21\lastx) edge (20\i)

(30\lastx) edge (31\i)
(31\lastx) edge (32\i)
(32\lastx) edge (30\i)
;
}

\path[-stealth]
(-.7,0) edge (i0)
(s) edge[loop right] ()
(210) edge[loop left] ()
(310) edge[loop left] ()
(320) edge[loop left] ()
(idots) edge (i7)
(ddots) edge[fault] (d7)
(20dots) edge (217)
(21dots) edge (207)
(30dots) edge (317)
(31dots) edge (327)
(32dots) edge (307)
;

\path[rounded corners,-stealth,draw, thick]
(d0) -- (-.8,2*\y) -- (-.8,8*\y) -- (s)
;
\path[rounded corners,draw, thick]
(200) -- (-.8,3*\y) -- (-.8,3*\y-1)
(300) -- (-.8,5*\y) -- (-.8,5*\y-1)
;



\end{tikzpicture}
\caption{The one-couter safety game~$\game_2$ for the proof of Lemma~\ref{lemma_lbresilienceinitialvertexsafety}.\ref{lemma_lbresilienceinitialvertexsafety_ocs}. Round vertices are in $V_0$, square ones in $V_1$, and disturbance edges are dashed. Player~$0$ wins if and only if $(s,\bot)$ is never visited. Vertices in the upper gray rectangle implement a modulo-$2$ counter while vertices in the lower rectangle implement a modulo-3 counter.}
\label{fig_safetyresillowerbounds_ocs}
\end{figure*}

The winning condition is defined such that Player~$0$ wins a play if and only if the state~$s$ is never reached.
Now, a play starting in the initial vertex of $\game_2$ proceeds as follows:
Player~$1$ either stays in the state~$i$ ad infinitum, and thereby allows Player~$0$ to win, or he eventually moves to some vertex of the form~$(c,A^n\bot)$.
Now, Player~$0$ has three choices, moving to $((2,0),A^n\bot)$, $((3,0),A^n\bot)$, or $(d,A^n\bot)$.
In the first case, there is only one continuation of the play  prefix, which results in a disturbance-free play that is winning for Player~$0$ if and only if $n \bmod 2 \neq 0$.
Similarly, in the second case, there is only one continuation of the play prefix, which results in a disturbance-free play that is winning for Player~$0$ if and only if $n \bmod 3 \neq 0$.
Finally, moving to $(d,A^n\bot)$ means that Player~$0$ wins if strictly less than $n$ disturbances occur in the continuation of the play prefix, but loses if $n$ disturbances occur.

We claim that the initial vertex has resilience~$6 = \prim{2}$. A $6$-resilient strategy for Player~$0$ moves from $(c,A^n\bot)$ to $(d,A^n\bot)$ if $n$ is a multiple of $6$.
Else, it moves to $((p_j,0), A^n\bot)$ for some $p_j \in \set{2,3}$ such that $n \bmod p_j \neq 6$, which always exists.
Applying the reasoning above implies that every play starting in the initial vertex, consistent with the strategy, and with at most $5$ disturbances is winning for Player~$0$.
Thus, the strategy is indeed $6$-resilient.

Now, consider an arbitrary strategy~$\sigma$ for Player~$0$. 
We show that it is not $7$-resilient, which is yields the desired result. 
To this end, consider the unique play prefix leading to $(c,A^6\bot)$, which is consistent with $\sigma$.
If $\sigma$ prescribes a move to some $((p_j,0),A^6\bot)$, then, as argued above, there is disturbance-free play that is consistent with $\sigma$, but losing for her.
The only other choice for $\sigma$ is to move to $(d,A^6\bot)$.
Then, as argued above as well, there is a play that is consistent with $\sigma$ with $6$ disturbances that is losing for her.
In both cases, we have shown that the strategy is indeed not $7$-resilient. 

The general case is obtained by having modulo counters in $\game_k$ for the first $k$ prime numbers instead of only the first two as in $\game_2$. 
Using the same reasoning as above for arbitrary $k$ instead of $k=2$ shows that the initial vertex of $\game_k$ has resilience~$\prim{k}$.

Finally, the number of states of the one-counter system inducing $\game_k$ is bounded by $\bigo(k^3)$.
	
\ref{lemma_lbresilienceinitialvertexsafety_pds}.)
We modify the one-counter safety game~$\game_k$ to obtain a pushdown safety game~$\game_k'$.
We use the stack alphabet~$\set{0,1}$, which allows us to interpret stack contents as binary encodings of natural numbers, with the least significant bit at the top of the stack.
In the following, we give an informal account of the structure of $\game_k'$ and leave the implementation by a pushdown system to the reader.
Here, we reuse the modulo-counters of $\game_k$ which forces that Player~$1$ to reach a stack height that is a multiple of $\prim{k}$, as he would lose otherwise.

Thus, Player~$1$ is initially forced to push a multiple of $\prim{k}$ $1$'s on the stack and then gives control to Player~$0$.
As the stack height is a multiple of $\prim{k}$, she can only go to a state~$d$ where all $0$'s are popped from the stack until the first $1$ is uncovered (note that initially there is no $0$ to pop).
If there is no such $1$, i.e., if the bottom of the stack is reached by removing $0$'s, then the play reaches a losing sink for Player~$0$.
Otherwise, if a $1$ is uncovered, then Player~$0$ only has a self-loop that leaves the stack unchanged, but there is also a disturbance edge that removes the topmost $1$ by a $0$ and hands back control to Player~$1$.
He can now push as many $1$'s as necessary to again reach a stack height that is a multiple of $\prim{k}$.

Now, if Player~$1$ never exceeds the stack height~$\prim{k}$, the stack always contains $\prim{k}$ bits when Player~$0$ gains control.
Assume now that Player~$0$ uses a strategy which moves to $d$ in that situation and uses the correct modulo counter to win in all other situations (as described in more detail above for $\game_k$).
Then, the stack contents reached at the positions where Player~$0$ gains control implement a binary counter with $\prim{k}$ bits that is decremented each time Player~$0$ gains control, starting with the value~$1^{\prim{k}}$.
Hence, as each decrement requires exactly one disturbance (and there are no others), the strategy described above is $(2^{\prim{k}}-1)$-resilient from the initial vertex.

On the other hand, $2^{\prim{k}}-1$ disturbances suffice to reach a stack containing only $0$'s at some configuration where Player~$0$ gains control. 
Then, the unique continuation of that play is losing for her. 

The only other choice for Player~$0$ is to enter a modulo counter at an \myquot{unsuitable} configuration, which also leads to a losing play with less than $2^{\prim{k}}$ disturbances.
Hence, Player~$0$ has no $2^{\prim{k}}$-resilient strategy from the initial vertex, i.e., it has indeed resilience~$2^{\prim{k}}-1$.

Finally, the number of states of the one-counter system inducing $\game_k$ is bounded by $\bigo(k^3)$.
\end{proof}

\section{Resilience in One-counter Safety Games}
\label{sec_ocs}
In this section, we show that one can compute the resilience of the initial vertex in a one-counter safety game in polynomial space, significantly improving the exponential space requirement derived in the previous section.

\begin{theorem}
\label{theorem_resiliencesafetyinpspace}
The following problem can be solved in polynomial space: \myquot{Given a one-counter safety game~$\game$ with initial vertex~$v_\initmark$, determine $r_\game(v_\initmark)$}.
\end{theorem}

To prove this result, we show that one can implement Algorithm~\ref{algorithm_ocssafety} in polynomial space if the underlying pushdown system is one-counter.
In this case, one can run the check \myquot{$v_\initmark \in \winreg_0(\arena_\rig, \safety(F)_\rig) $} in Line~1 in polynomial space due to Lemma~\ref{lemma_riggedpushdowngameseffective}, and can implement the counter in Line~3 in polynomial space, as the upper bound~$b(\pds)$ is exponential (see the definition on Page~\pageref{bounddef}).
It remains to show that one can check
in polynomial space, for a given $k\le b(\pds)$, if
$v_\initmark\in\winreg_1(\arena_\rig, \safety(F)_\rig \cup R_{\ge k})$ holds.
In the rest of this section we show that this is indeed possible.

Fix, the rigged game~$\game_k = (\arena_\rig, \safety(F)_\rig \cup R_{\ge k})$ for some $k \le b(\pds)$ with $\arena_\rig = (V', V_0', V_1', E', \emptyset)$, with initial vertex~$v_\initmark$, where $\pds$ is the OCS underlying the original game~$\game$ that induces $\game_k$.
We show that the existence of winning strategies for Player~$1$ in $\game_k$
can be witnessed by a finite graph structure, as follows.

A \emph{strategy graph} for $\game_k$ is a tuple~$(V^\sgmark, E^\sgmark, \sgrankr^\sgmark, \sgrankd^\sgmark)$ with $\sgrankr^\sgmark\colon V^\sgmark \rightarrow \set{0,\ldots, k-1}$ and $\sgrankd^\sgmark\colon V^\sgmark \rightarrow \set{0, \ldots, \size{V^\sgmark}}$ such that the following properties are satisfied:
\begin{enumerate}
	\item\label{graphproperty:graph} $(V^\sgmark, E^\sgmark)$ is a directed graph with $V^\sgmark \subseteq V'$, $E^\sgmark \subseteq E'$, $v_I \in V^\sgmark$, and $\sh(v) \le (2k)^{\size{Q}^2}$ for all $v \in V^\sgmark$. Note that $(2k)^{\size{Q}^2}$ is exponential in the size of the pushdown system~$\pds$ underlying $\game$, even though $k \le b(\pds)$ may itself be exponential.
	
	\item\label{graphproperty:strategy0} For all $v \in (V^\sgmark \cap V_0') \setminus F$ and all $(v,v') \in E'$, we have $(v,v') \in E^\sgmark$.
	 
	\item\label{graphproperty:strategy1} For all $v \in (V^\sgmark \cap V_1') \setminus F$ there is a unique outgoing edge~$(v,v') \in E'$ with $(v,v') \in E^\sgmark$.
	
	\item\label{graphproperty:valuesr} For all $(v,v') \in E^\sgmark$, we have $\sgrankr^\sgmark(v) \ge \sgrankr^\sgmark(v')$ with strict inequality if $v \in D$.
	
	\item\label{graphproperty:valuesd} For all $(v,v') \in E^\sgmark$, we have $\sgrankd^\sgmark(v) > \sgrankd^\sgmark(v')$. 
	
\end{enumerate}

\begin{lemma}
\label{lemma_stratgraphcharacterizeswinninggamek}
Player~$1$ wins $\game_k$ from $v_\initmark$ if and only if there exists a strategy graph for $\game_k$.
\end{lemma}

To simplify the proof, we transform $\game_k$ into a game~$\game_k'$ where all reachable vertices in $F$ are sinks of stack height zero. 
To do this, we replace all outgoing (standard and disturbance) edges of vertices~$(q,A^n\bot) \in F$ with $n>0$ by an edge to $(q,A^{n-1}\bot)$ (which is also in $F$) and the all outgoing (standard and disturbance) edges of vertices~$(q,\bot) \in F$ by an edge to a sink vertex~$(q_f,\bot)$, where $q_f$ is a fresh state. 
Then, $\game_k'$ is the game played in the modified arena with winning condition~$\safety(\set{q_f})_\rig \cup R_{\ge k}$.
Intuitively, once a vertex in $F$ is reached in the modified arena, the players no longer have strategic choices; 
instead, the stack is emptied (without simulating any disturbances) and the unsafe sink vertex~$(q_f,\bot)$ is reached.

It is straightforward to verify that we have $v \in \winreg_i(\game_k)$ if and only $v \in \winreg_i(\game_k')$ for every vertex of $\arena_\rig$ and $i \in \set{0,1}$ by transferring winning strategies between the games.
So, in the following, we assume without loss of generality, that the only vertices of $\game_k$ in $F$ that are reachable from the initial vertex are sinks of stack height zero.
In this situation, a play can no longer simulate a disturbance edge once a vertex in $F$ has been reached. 

To prove Lemma~\ref{lemma_stratgraphcharacterizeswinninggamek}, we show that if Player~$1$ wins $\game_k$ with some arbitrary winning strategy, then also with a winning strategy that can be turned into a strategy graph.
To simplify our notation, given a strategy~$\tau$, let $\maxheight(\tau) = \sup_v \sh(v)$, where $v$ ranges over all vertices reachable by a play prefix starting in $v_\initmark$ that is consistent with $\tau$, i.e., $\maxheight(\tau)$ is the maximal stack height visited by a play that is starting in the initial vertex and consistent with $\tau$.
Using this, we show that Player~$1$ wins $\game_k$ from $v_\initmark$ if and only if he has a positional winning strategy from $v_\initmark$ with $\maxheight(\tau) \le (2k)^{\size{Q}^2}$.  
The latter can then be transformed into a strategy graph.

We only have to consider the implication from left to right, as the other one is trivial.
Let Player~$1$ win $\game_k$ from $v_\initmark$, i.e., he has a winning strategy~$\tau$ for $\game_k$ from $v_\initmark$.
We proceed in two steps:
First, We turn $\tau$ in a positional winning strategy~$\tau'$ from $v_\initmark$ (Lemma~\ref{lemma_positionalstategies}).
Then, we turn $\tau'$ into a positional winning strategy~$\tau''$ with $\maxheight(\tau'') \le (2k)^{\size{Q}^2}$ (Lemma~\ref{lemma_positionalstategiesexponentialstackheightgamek}).

For the first step, we generalize a standard argument for turning an arbitrary, not necessarily positional, winning strategy~$\tau$ in a reachability game into a positional one: 
At a vertex~$v \notin F$, consider all play prefixes that are consistent with $\tau$ and end in $v$, and mimic the move $\tau$ prescribes for a longest one (call it $\rep(v)$). 
The resulting strategy~$\tau'$ is obviously positional and winning as every play consistent with $\tau'$ and ending in some $v \notin F$ can be shown to be at most as long as the play~$\rep(v)$ whose moves are mimicked to define $\tau'(v)$.
Here, we have to refine this argument to ensure that the resulting strategy~$\tau'$ still simulates at most $k-1$ disturbances during each play.

\begin{lemma}
\label{lemma_positionalstategies}
If Player~$1$ wins $\game_k$ from $v_\initmark$ then he has a positional winning strategy for $\game_k$ from $v_\initmark$.
\end{lemma}

\begin{proof}
Assume a winning strategy $\tau$ for Player~$1$ from $v_\initmark$.
Let us call a play prefix~$v_0 \cdots v_j$ \emph{unsettled} if it starts in $v_\initmark$, is consistent with $\tau$, and 
no strict prefix contains a vertex in the target $F$.
Notice that there must be a uniform bound~$\ell \in \omega$ such that $\size{w} < \ell$ for every unsettled $w$.
Indeed, if there was no such bound, then it is possible to arrange an infinite set of arbitrarily long play prefixes not visiting $F$ into an infinite finitely branching tree. 
By König's Lemma, this tree has an infinite path which corresponds to an infinite play starting in $v_\initmark$, consistent with $\tau$, but not containing a vertex in $F$,
which contradicts the assumption that $\tau$ is winning.

Given an unsettled prefix~$w$, let $\val(w) = d \cdot \ell + \size{w}$ where $d$ is the number of simulated disturbances during $w$. 
Let $U(v)$ for $v \in V'$ denote the set of unsettled play prefixes ending in $v$. 
Further, for every $v \in V$ with non-empty $U(v)$ let $\rep(v)$ be an element from $U(v)$ such that $\val(\rep(v)) \ge \val(w)$ for all $w \in U(v)$.
Such an element exists, as the $\val(w)$ for $w \in U(v)$ are bounded by $ k \cdot\ell -1$:
Each unsettled prefix is consistent with the winning strategy~$\tau$, which implies that it simulates at most $k-1$ disturbance edges, and its length is bounded by $\ell-1$, as argued above.

Based on this we define the positional strategy~$\tau'$ via $\tau'(v) = \tau(\rep(v))$ if $\rep(v)$ is defined and $\tau'(v) = v'$ for some arbitrary successor~$v'$ of $v$ if $\rep(v)$ is undefined (note that it suffices to define $\tau'(v)$ for $v \in V_1$ to define a positional strategy for Player~$1$).
We claim that $\tau'$ is winning from $v_\initmark$.
To this end, let $\rho = v_0v_1v_2 \cdots$ start in $v_\initmark$ and be consistent with $\tau'$. 
We need to show that it visits a vertex in $F$ and that it simulates at most $k-1$ disturbance edges.

A simple induction shows that every length-$j$ prefix $v_0 \cdots v_{j-1}$ that does not visit $F$ must satisfy that
\begin{equation}
\val(v_0 \cdots v_j) \le \val(\rep(v_j))
\hfill \label{eq}
\end{equation}
The induction start $j = 0$ is trivial, as we have $v_0 = v_\initmark$ and $v_\initmark \in U(v_\initmark)$, which implies $\val(v_\initmark ) \le \val(\rep(v_\initmark))$ as required.

For the induction step, consider some $j > 0$ such that $v_0 \cdots v_{j-1}$ does not visit $F$. 
The induction hypothesis yields $\val(v_0 \cdots v_{j-1}) \le \val(\rep(v_{j-1}))$. 
Let $ \rep(v_{j-1}) = wv_{j-1}$, which is consistent with $\tau$.
If $v_{j-1} \in V_0'$, then $wv_{j-1}v_j$ is consistent with $\tau$ as well, as it is Player~$0$'s turn at $v_{j-1}$.
Similarly, if $v_{j-1} \in V_1'$, then we have 
\[v_j = \tau'(v_0 \cdots v_{j-1}) = \tau(\rep(v_{j-1})) = \tau(wv_{j-1}). \]
Hence, $wv_{j-1}v_j$ is again consistent with $\tau$.
Furthermore, $wv_{j-1}$ does not contain a vertex in $F$, as $v_{j-1}$ is not in $F$ (recall that vertices in $F$ are sinks).
Thus, we conclude that $wv_{j-1}v_j$ is unsettled, which implies $\val(wv_{j-1}v_j) \le \val(\rep(v_j))$, by our definition of $\rep(v_j)$.
To finish the induction step,
let $x = 1$ if $ v_j \in D$, i.e. a disturbance edge is simulated, and $x = 0$ otherwise.
Then, we have \begin{align*}
\val(v_0 \cdots v_j) ={}& \val(v_0 \cdots v_{j-1}) + x\cdot \ell + 1 \\
\le{}& \val(\rep(v_{j-1})) + x\cdot \ell + 1 \\
={}& \val(wv_{j-1}) + x\cdot \ell + 1 \\
={}& \val(wv_{j-1}v_j) \le \val(\rep(v_j)).
\end{align*}

Applying Equation~\ref{eq}, we can show that $\rho$ is indeed winning.
First, towards a contradiction, assume $\rho$ does not visit a vertex in $F$.
Then, Equation~\ref{eq} is applicable to every prefix $v_0 \cdots v_j$ and we thus obtain for every $j>0$, that
\[
j+1 = \size{v_0 \cdots v_j} \le \val(v_0 \cdots v_j) \le \val(\rep(v_j))
\le k \cdot\ell -1
\]
which is a contradiction as the term on the right is constant.

Second, again towards a contradiction, assume that $\rho$ simulates at least $k$ disturbance edges.
Then, let $j$ be minimal such that the prefix~$v_0 \cdots v_j$ simulates exactly $k$ disturbance edges.
As vertices in $F$ are sinks, and therefore have no outgoing edges simulating disturbance edges, Equation~\ref{eq} is applicable to $v_0 \cdots v_j$ and we obtain that
\[
k\cdot \ell \le \val(v_0 \cdots v_j) \le \val(\rep(v_j))
\]
which is impossible as $\val(\rep(v_j)) \le k \cdot\ell -1$.
Consequently, $\rho$ visits $F$ and simulates at most $k-1$ disturbance edges. As $\rho$ was an arbitrary play consistent with $\tau'$, this strategy is indeed winning.
\end{proof}

The second step of our construction is to bound the stack height reached by plays consistent with the winning strategy (while preserving positionality).
To this end, we generalize a classical argument for pushdown safety games:
In such games, Player~$1$, who has a reachability objective, has a positional winning strategy~$\tau$ from $v_\initmark$ with exponentially bounded~$\maxheight(\tau)$, if he wins at all from $v_\initmark$.
This is typically proven by a \myquot{hill-cutting} argument
\cite{DBLP:journals/jcss/BohmGJ14,Val1973} 
showing that a winning strategy exceeding this bound can be turned into one of smaller maximal stack height by removing infixes of plays that increase the stack without reaching states that have not been reached at smaller stack height already. 
Here, we again have to generalize this argument to additionally ensure that the number of simulated disturbances remains bounded by $k-1$. This is done using ``summarizations'' of paths in pushdown systems (see e.g.~\cite{RepsHS95,HMM2016}) that take the number of disturbances into account.

\begin{lemma}
\label{lemma_positionalstategiesexponentialstackheightgamek}
If Player~$1$ wins $\game_k$ from $v_\initmark$ then he has a positional winning strategy from $v_\initmark$ with $\maxheight(\tau) \le (2k)^{\size{Q}^2}$.  
\end{lemma}

\begin{proof}
By Lemma~\ref{lemma_positionalstategies}
we can pick a positional strategy~$\tau$ for Player~$1$ that is winning $\game_k$ from $v_\initmark$.
We show how to turn this into a winning strategy that satisfies the claim.

Notice first that $\maxheight(\tau)$ must be finite.
Indeed, if it is unbounded, then for every $n \in \omega$ there is a play prefix~$w_n$ starting in $v_\initmark$, consistent with $\tau$, and ending in a vertex of stack height~$n$.
As the stack height is increased by at most one during each move, we have $\size{w_n} \ge n$.
Furthermore, as vertices in $F$ are sinks, these play prefixes can be assumed to not contain a vertex in $F$.
The prefixes~$w_n$ can be arranged in an infinite finitely branching tree. 
By König's Lemma, this tree has an infinite path, which corresponds to an infinite play starting in $v_\initmark$, consistent with $\tau$, but not visiting a vertex in $F$.
This contradicts $\tau$ being a winning strategy. 

It suffices to show that
if $\maxheight(\tau) > (2k)^{\size{Q}^2}$, then $\tau$ can be turned into a positional winning strategy~$\tau'$ from $v_\initmark$ with strictly smaller maximal stack height. 
%
For the sake of readability, we will identify a stack content~$A^n\bot$ of the one-counter system underlying $\game_k$ by the number~$n \in \omega$.
Hence, vertices of $\game_k$ are from now on denoted by $(q,n)$ with $n \in \omega$.

Let $R$ denote the set of vertices reachable from $v_\initmark$ via play prefixes that are consistent with $\tau$.
For $(q,n) \in R$ with $n > 0$ let $H(q,n)$ be the set of vertices of the form~$(q',n-1)$ reachable from $(q,n)$ via a play prefix~$(q,n)(q_1,n_1) \cdots (q_j,n_j)(q',n-1)$ that is consistent with $\tau$ and such that $n_{j'} \ge n$ for every $j' \in \set{1, \ldots, j}$, i.e., the last vertex of the play prefix is the first time the stack height along the play prefix is strictly smaller than $n$.
We call such a play prefix a \emph{hill} from $(q,n)$ to $(q',n-1)$.

For all $n > 0$ define the partial function~$h_n \colon Q \rightarrow 2^Q$ that maps $q$ to $H(q,n)$ whenever $(q,n) \in R$, and else leaves $h_n(q)$ undefined.
Similarly, define the partial function~$d_n \colon Q \times Q \rightarrow \set{0, \ldots, k-1}$ by mapping each pair~$(q,q')$ with $q' \in H(q,n)$ to the maximal number of disturbances simulated during any hill from $(q,n)$ to $(q',n-1)$.
This value is bounded by $k-1$, as each hill is part of a play that is consistent with $\tau$.
For $(q,q')$ with $q' \notin H(q,n)$, we leave $d_n(q,q')$ undefined.

There are at most $(2k)^{\size{Q}^2}$ many different pairs of such functions $h_n$ and $d_n$.
Hence, if $R$ contains a vertex~$(q,n)$ with $n > (2k)^{\size{Q}^2}$, then there are $0 < n_\ell < n_u$ such that $h_{n_\ell} = h_{n_u}$ and $d_{n_\ell} = d_{n_u}$.
Let $s = n_u -n_\ell$.
We define the positional strategy~$\tau'$ via $\tau'(q,n) = \tau(q,n)$, if $n < n_\ell$ and $\tau'(q,n) = \tau(q, n + s)$ if $n \ge n_\ell$ (recall that it suffices to define $\tau'(v)$ for every $v \in V_1'$ to define a positional strategy~$\tau'$). 

We claim that $\tau'$ is still winning for Player~$1$ from $v_\initmark$ in $\game_k$.
To this end, consider an arbitrary play 
\[\rho' = (q_0, n_0)(q_1, n_1)(q_2, n_2) \cdots\] that starts in $v_\initmark$ and is consistent with $\tau'$.
We need to show that it visits $F$ and simulates at most $k-1$ disturbance edges.

If every $n_j$ is strictly smaller than $n_\ell$, then $\rho'$ is also consistent with $\tau$, as only the first case of the definition of $\tau$ is applied.
Hence, it is winning for Player~$1$, as $\tau$ is a winning strategy from $v_\initmark$.

It remains to consider the case where $\rho'$ reaches stack height~$n_\ell$.
Here, we turn $\rho'$ into a play~$\rho$ starting in $v_\initmark$ and consistent with $\tau$, which implies that $\rho$ visits $F$ and simulates at most $k-1$ disturbance edges.
Using the relation between $\rho$ and $\rho'$, we argue that the $\rho'$ is also winning. 

The following remark is useful throughout our argument and follows immediately from the fact that at stack heights~$n \ge n_\ell$, $\tau'$ mimics the behavior of $\tau$ at stack height~$n+s$. 

\begin{remark}
\label{remark_copycatstrategy}	
	Let $j$ and $j'$ be positions of $\rho'$ such that $n_j = n_\ell$ and $n_{j''} \ge n_\ell$ for every $j'' \in \set{j+1, \ldots, j'}$, i.e., the infix between positions~$j$ and $j'$ starts at stack height~$n_\ell$ and never reaches a smaller stack height. Then, $(q_j, n_j+s) \cdots (q_{j'+1}, n_{j'+1}+s)$ is consistent with $\tau$ (note the $+1$!).
\end{remark}

We inductively construct $\rho$ by defining a sequence~$(w_m)_{m \in \omega}$ of strictly increasing prefixes whose limit is $\rho$.
To define this sequence, we simultaneously construct a sequence~$(j_m)_{m \in \omega}$ of strictly increasing positions of $\rho'$.
During the construction, we satisfy the following invariant: 
Each $w_m$ is consistent with $\tau$, ends in $(q_{j_m}, n_{j_m})$ where $n_{j_m}$ is strictly smaller than $n_\ell$, and $w_m$ simulates at least as many disturbances as $(q_0, n_0) \cdots (q_{j_m}, n_{j_m})$.

We start with $j_0 = 0$ and $w_0 = (q_0, n_0) = v_\initmark$, which satisfies the invariant due to our choice of $n_\ell$ being greater than zero.
We define $w_m$ and $j_m$ for $m >0$, based on $w_{m-1}$ and $j_{m-1}$, as follows.
Due to the invariant, $w_{m-1}$ ends in $(q_{j_{m-1}}, n_{j_{m-1}})$ with $n_{j_{m-1}} < n_\ell$ and is consistent with $\tau$.

We consider two cases.
In the first,  if $(q_{j_{m-1}+1}, n_{j_{m-1}+1})$, the next vertex after $(q_{j_{m-1}}, n_{j_{m-1}})$ in $\rho'$, satisfies $n_{j_{m-1}+1} < n_\ell$, then define $w_{m} = w_{m-1}(q_{j_{m-1}+1}, n_{j_{m-1}+1})$ and $j_m = j_{m-1}+1$.
Note that the move from $(q_{j_{m-1}}, n_{j_{m-1}})$ to $(q_{j_{m-1}+1}, n_{j_{m-1}+1})$ is consistent with $\tau$, as it is either Player~$0$'s turn or the first case of the definition of $\tau'$ is applied (which mimics $\tau$) due to our invariant. 
Hence, $w_m$ is again consistent with $\tau$. 
Similarly, the requirement on the number of simulated disturbances is satisfied as the same edge is used to extend both play prefixes.

In the second case, we have $n_{j_{m-1}+1} \ge n_\ell$, which implies $n_{j_{m-1}+1} = n_\ell$, as the stack height can increase by at most one during every transition.

We claim there is some $j > j_{m-1}+1$ such that $n_j = n-1$. 
Towards a contradiction, assume there is no such $j$. 
Then, Remark~\ref{remark_copycatstrategy} is applicable to every pair $(j_{m-1}+1,j)$ with $j > j_{m-1}+1$. 
This yields an infinite play~$\rho_c$ equal to 
\[
(q_{j_{m-1}+1}, n_{j_{m-1}+1}+s) 
(q_{j_{m-1}+2}, n_{j_{m-1}+2}+s) 
(q_{j_{m-1}+3}, n_{j_{m-1}+3}+s)\cdots
\]
that is consistent with $\tau$.
The play prefix~$w_{m-1}$ starts in $v_\initmark$, is consistent with $\tau$, and ends in $(q_{j_{m-1}}, n_{j_{m-1}})$.
Further, the move from $(q_{j_{m-1}}, n_{j_{m-1}})$ to $(q_{j_{m-1}+1}, n_{j_{m-1}+1})$ in $\rho'$ is consistent with $\tau'$ and therefore also with $\tau$, as $n_{j_{m-1}} < n_\ell$ by our invariant.
Thus, we have shown $(q_{j_{m-1}+1}, n_{j_{m-1}+1}) \in R$, i.e., there is a play prefix~$w_c$ starting in $v_\initmark$, consistent with $\tau$, and ending in $(q_{j_{m-1}+1}, n_{j_{m-1}+1})$.
Altogether, we can combine $w_c$ and $\rho_c$ into an infinite play starting in $v_\initmark$ and consistent with $\tau$ that has $\rho_c$ as suffix.
Now, $\rho_c$ contains by construction no vertex of stack height zero.
As vertices in $F$ are sinks of stack height zero, the combined play can  not visit $F$.
This contradicts the assumption that $\tau$ is winning from~$v_\initmark$. 

Thus, let $j > j_{m-1}+1$ be minimal such that $n_j = n-1$. 
Applying Remark~\ref{remark_copycatstrategy} for $j_{m-1}+1$ and $j-1$ shows that 
\[
w = (q_{j_{m-1}+1}, n_{j_{m-1}+1}+s) \cdots (q_{j}, n_{j}+s)
\]
is a hill from $(q_{j_{m-1}+1}, n_{j_{m-1}+1}+s) = (q_{j_{m-1}+1}, n_u)$ to $(q_{j}, n_{j}+s) = (q_{j}, n_u-1)$.
Hence, by the choice of $n_\ell$ and $n_u$ there is also a hill~$w'$ from 
$(q_{j_{m-1}+1}, n_\ell)$ to $ (q_{j}, n_\ell-1)$ that has at least as many simulated disturbances as $w$.

We obtain $w_m$ from $w_{m-1}$ by appending $w'$ and define $j_m = j$.
The requirement on the stack height~$n_{j_m}$ is satisfied by our choice of $j_m = j$ while $w_m$ is consistent with $\tau$, as $w_{m-1}$, the move from $(q_{j_{m-1}}, n_{j_{m-1}})$ (the last vertex of $w_{m-1}$) to $(q_{j_{m-1}+1}, n_{j_{m-1}+1})$ (the first vertex of $w'$), and $w'$ are all consistent with $\tau$.
The requirement on the number of simulated disturbances is satisfied, as $(q_{j_{m-1}+1}, n_{j_{m-1}+1}) \cdots (q_{j}, n_{j})$ simulates the same number of disturbances as $w$, which is at most the number of disturbances simulated by $w'$. 

Consider the resulting play~$\rho$, which is by construction winning for Player~$1$ and consequently simulates at most $k-1$ disturbances.
An inductive application of the invariant above shows that $\rho'$ therefore also simulates at most $k-1$ disturbances.
Furthermore, $\rho$ visits a vertex in $F$, which has stack height zero. 
When such a vertex is added during the inductive construction described above, then only in the first case (when $n_{j_{m-1}+1} < n_\ell$) and only because the same vertex appears in $\rho'$, i.e., $\rho'$ visits $F$ as well.
Hence, $\rho'$ is indeed winning for Player~$1$.

To conclude, we have to show $\maxheight(\tau') < \maxheight(\tau)$.
An induction on $n$ shows that if $(q,n)$ is reachable from $v_\initmark$ by a play prefix that is consistent with $\tau'$, then:
\begin{itemize}
	
	\item If $n \le n_\ell$, then $(q,n)$ is reachable from $v_\initmark$ by a play prefix that is consistent with $\tau$.
	
	\item If $n > n_\ell$, then $(q,n+s)$ is reachable from $v_\initmark$ by a play prefix that is consistent with $\tau$.

\end{itemize}
This implies $\maxheight(\tau') + s \le \maxheight(\tau)$, which yields the desired bound due to $s >0$.
\end{proof}

A positional strategy as in Lemma~\ref{lemma_positionalstategiesexponentialstackheightgamek} is essentially a strategy graph.
So, we have proven Lemma~\ref{lemma_stratgraphcharacterizeswinninggamek}: The existence of strategy graphs for $\game_k$ captures Player~$1$ winning $\game_k$. 

\begin{proof}[Proof of Lemma~\ref{lemma_stratgraphcharacterizeswinninggamek}]
Let Player~$1$ win $\game_k$ from $v_\initmark$.
Then, Lemma~\ref{lemma_positionalstategiesexponentialstackheightgamek} yields a positional winning strategy~$\tau$ for him from $v_\initmark$ with $\maxheight(\tau) \le (2k)^{\size{Q}^2}$.
We turn $\tau$ into a strategy graph for $\game_k$ by defining
\begin{itemize}

	\item $V^\sgmark$ to be the set of vertices visited by plays starting in $v_\initmark$ that are consistent with $\tau$,
	
	\item $E^\sgmark$ to be the set of edges traversed by these plays (ignoring the self-loops at vertices in $F$), 
	
	\item $\sgrankr^\sgmark(v)$ to be the maximal number of disturbance edges simulated on plays starting in $v$ that are consistent with $\tau$, and
	
	\item $\sgrankd^\sgmark(v)$ to be the maximal length of a play prefix starting in $v$, being consistent with $\tau$, and the last vertex (but no other) being in $F$.

\end{itemize}
It is straightforward to prove that $(V^\sgmark, E^\sgmark, \sgrankr^\sgmark, \sgrankd^\sgmark)$ satisfies all properties required of a strategy graph for $\game_k$.

Conversely, assume there is a strategy graph~$(V^\sgmark,E^\sgmark,\sgrankr^\sgmark,\sgrankd^\sgmark)$ for $\game_k$.
We turn it into a positional winning strategy~$\tau$ for Player~$1$ from $v_\initmark$.
Let $v \in V_1'$.
If $v \in V^\sgmark \setminus F$, then there is a unique outgoing edge~$(v,v') \in E^\sgmark \cap E'$ due to Property~\ref{graphproperty:strategy1} of the strategy graph definition.
Then, we define $\tau(v) = v'$.
Otherwise, i.e., if $v \notin V^\sgmark \setminus F$, then define $\tau(v)$ to be an arbitrary successor of $v$ in $\arena_\rig$.
We claim that $\tau$ is indeed a winning strategy for Player~$1$ for $\game_k$ from $v_\initmark$.

To this end, let $\rho = (v_0, 0) (v_1, 0) (v_2, 0) \cdots$ be a play starting in $v_\initmark$ that is consistent with $\tau$.
We need to show that $\rho$ is winning for Player~$1$, i.e., that it visits $F$ and contains at most $k-1$ simulated disturbance edges.

An induction applying the definition of $\tau$ and Property~\ref{graphproperty:strategy0} of the strategy graph definition shows that if $v_0 \cdots v_{j-1}$ does not contain a vertex from $F$, then $v_0 \cdots v_{j}$ is a path through the graph~$(V^\sgmark, E^\sgmark)$. 
Hence, we have $ \sgrankd^\sgmark(v_0) > \sgrankd^\sgmark(v_0) > \cdots > \sgrankd^\sgmark(v_{j-1}) > \sgrankd^\sgmark(v_j)$ by Property~\ref{graphproperty:valuesd}.
As the range of $\sgrankd^\sgmark$ is finite, this yields an upper bound on the length of prefixes of $\rho$ that do not visit $F$, which implies that $\rho$  contains a vertex of $F$. 
Hence, let $j$ be the minimal position of $\rho$ with $v_j \in F$. 
As vertices in $F$ are sinks, no disturbance edges are simulated in $\rho$ after position~$j$.
Due to Property~\ref{graphproperty:valuesr}, we have $ \sgrankr^\sgmark(v_0) \ge \sgrankd^\sgmark(v_0) \ge \cdots \ge \sgrankr^\sgmark(v_{j-1}) \ge \sgrankr^\sgmark(v_j)$ with strict inequality whenever a disturbance edge is simulated, as $v_0 \cdots v_j$ is a path through $(V^\sgmark, E^\sgmark)$ as argued above. 
Hence, as the range of $\sgrankr^\sgmark$ has at most $k$ elements, there are at most $k-1$ simulated disturbances in $\rho$ before position~$j$ and none afterwards, as argued above.
Altogether, $\rho$ visits $F$ and contains at most $k-1$ simulated disturbance edges, i.e., it is indeed winning for Player~$1$ in $\game_k$. 
\end{proof}

Hence, it remains to prove that we can decide the existence of strategy graphs in polynomial space.
Here, we use the fact that $k$ is at most $b(\pds) \in \bigo(2^{\size{\pds}^2})$, where $\pds$ is the pushdown system underlying the game inducing $\game_k$, to guess and verify a strategy graph in polynomial space.

\begin{lemma}
\label{lemma_stratgraphexistenceinpspace}
The following problem is in $\pspace$: \myquot{Given a one-counter safety game~$\game$ induced by a PDS~$\pds$ and $k \le b(\pds)$ (encoded in binary), is there a strategy graph for $\game_k$?}.
\end{lemma}

\begin{proof}
Notice that all defining conditions of strategy graphs are local, and can be verified
for a vertex $v=(q,n)$ if the values of $\sgrankr^\sgmark(v')$ and $\sgrankd^\sgmark(v')$ are known for all direct neighbors,
which have the form $(q',n')$ with $n'\in\set{n-1,n,n+1}$.
A strategy graph can therefore be guessed and verified on the fly,
keeping in memory these values for vertices in $Q \times \set{n, n+1, n+2}$ while incrementing $n$ from $0$ to $(2k)^{\size{Q}^2}$.
This requires polynomial space, both for the labeling of the vertices (as the numbers are at most exponential in the size of the input) and for the counter.
\end{proof}

While we consider one-counter systems with unit updates, i.e., each transition changes the counter value by at most one, our results are also applicable to one-counter systems where each transition updates the counter by some integer (encoded in binary).
Such \emph{binary updates} can be simulated by unit updates, albeit with an exponential blowup.
Hence, the algorithm above computes the resilience of the initial vertex of a one-counter safety game with binary updates in exponential space.
A matching lower bound follows from the $\expspace$ hardness of solving disturbance-free one-counter safety games with binary updates~\cite{Hunter15}.

\section{Beyond Safety: Reachability Games with Disturbances}

Up to now, we were concerned with pushdown safety games with disturbances and have shown that they provide a rich model with interesting properties that go beyond the rather straightforward case of finite safety games with disturbances.
Nevertheless, there are many other winning conditions that can be studied in pushdown games with disturbances.
Probably the simplest class of winning condition besides safety conditions are reachability conditions:
Given a set~$F \subseteq V$ of vertices, the reachability condition~$\reach(F) = \set{ v_0v_1v_2 \cdots \mid v_j \in F \text{ for some }j \in \omega }$ requires to visit $F$ at least once.
While safety and reachability conditons are dual, games with disturbances are asymmetric.
Thus, we cannot directly transfer results for safety to reachability games and vice versa.

\subsection{Resilience in Pushdown Reachability Games}

Many of our results trivially carry over to reachability conditions while others can be recovered with some more effort.
We begin by claiming that, in contrast to safety games, all possible resilience values can be realized in reachability games. 

\begin{lemma}
\label{lemma_allresiliencevaluesinreachability}
All possible resilience values~$\alpha \in \omega+2$ can be realized in a one-counter reachability game that has vertices of resilience~$\omega$ with and without a uniform witness.
\end{lemma}

\begin{proof}
Consider the one-counter reachability game~$\game$ presented in Figure~\ref{fig_reachresilvalues} where the reachability condition is induced by the doubly-lined vertices, i.e., Player~$0$ wins if and only if a doubly-lined vertex is visited.
For every $\alpha \in \omega +2$, there is a vertex~$v$ with $r_\game(v) = \alpha$, where the lower vertex of resilience~$\omega$ has a uniform witness while the upper row of vertices does not, for reasons that are analogous to the ones presented in Example~\ref{example_safetyallpossibleresiliencevalues}:
Essentially, the upper row of vertices implements the fresh vertex~$v$ to obtain $\game'$. 
\begin{figure*}
\centering
\begin{tikzpicture}[ultra thick,label distance=-1.2mm]

\foreach \i in {0,1,...,9}{
\node[p0s,label=below right:$\i$] (l\i) at (\i*1.5,-1.5) {};
\node[p0s,double,label =below right:$\omega+1$] (a\i) at (\i*1.5,-2.6) {};
\node[p0s,label=below right:$\omega$] (o\i) at (\i*1.5,-.4) {};
}

\node (adots) at (15,-2.6) {$\cdots$};
\node 	  (ldots) at (15,-1.5) {$\cdots$};
\node 	  (odots) at (15,-.4) {$\cdots$};

\node[p0s,label=below right:$\omega$] (b) at (0,-3.9) {$$};

\foreach \i in {0,...,9}{
\path[-stealth]
(o\i) edge[] (l\i);
}
\foreach \i in {1,...,9}{
\path[-stealth]
(l\i) edge[bend left=0] (a\i);
}

\foreach \i [remember=\i as \lasti (initially 0)]in {1,...,9}{
\path[-stealth]
(o\lasti) edge (o\i);
}

\foreach \i [remember=\i as \lasti (initially 9)]in {8,...,0}{
\path[-stealth]
(a\lasti) edge (a\i)
(l\lasti) edge[fault] (l\i)
;
}

\path[-stealth]
(adots) edge (a9)
(o9) edge (odots)
(ldots) edge[fault] (l9)
(b) edge (a0)
(a0) edge[loop left] ()
(b) edge[fault,loop left] ()
(l0) edge[loop left] ()
;	
\end{tikzpicture}
\caption{A one-counter reachability game with all possible resilience values (depicted as labels below vertices). All vertices belong to $V_0$ and disturbance edges are dashed.}
\label{fig_reachresilvalues}
\end{figure*}
\end{proof}

Thus, let us consider the problem of determining the resilience of the initial vertex of a pushdown reachability game.
Lemma~\ref{lemma_riggedpushdowngameseffective} and Lemma~\ref{lemma_riggedpushdowngameseffective_k} are formulated for parity games, and therefore hold in particular for reachability games, as these are subsumed by parity games.
Thus, we can determine whether the initial vertex of a pushdown reachability game has resilience~$\omega+1$, $\omega$ with a uniform witness, or $k$, for a fixed $k$.
However, as the value could also be $\omega$ without a uniform witness (for which we have no characterization) and as we have no upper bound on possible finite values, we do not obtain a complete algorithm determining the resilience of the initial vertex of every game.

Such an upper bound would immediately yield a complete algorithm similar to Algorithm~\ref{algorithm_ocssafety}:
one just has to add a line checking whether the resilience is $\omega$ (and returning that result) and if the resilience is not equal to some $k$ below the upper bound, then the algorithm returns \myquot{$\omega$ without uniform witness}.

Finally, the examples witnessing the lower bounds presented in Lemma~\ref{lemma_lbresilienceinitialvertexsafety} can be turned into reachability games.
Intuitively, one uses disturbances to push content on the stack, i.e., replaces the moves of Player~$1$ by disturbance edges.

\subsection{Optimal Strategies in One-counter Reachability Games}

In this subsection, we build a bridge 
between resilience in pushdown safety games and a classical problem in the theory of infinite games on infinite arenas:
computing optimal strategies in pushdown reachability games, i.e., winning strategies that reach a fixed set~$F$ of target states in the least number of steps possible. 
This problem has been first studied by Cachat~\cite{Cachat02} and recently been revisited by Carayol and Hague~\cite{CarayolH18}.
Here, we first present a connection between both problems and then prove a, to the best of our knowledge, novel result about optimal strategies in one-counter games.

Fix a reachability game~$\game= (\arena, \reach(F))$ with a distur\-ban\-ce-free finitely branching arena~$\arena = (V,V_0, V_1, E, \emptyset)$.
Given a play~$\rho = (v_0,0)(v_1,0)(v_2,0) \cdots$, let 
\[\val_\game(\rho) = \min \set{j \in \omega \mid v_j \in F}\]
 where we define $\min \emptyset = \omega+1$ for technical convenience.
Hence, $\val_\game(\rho)$ is the minimal position in $F$, if $v_0v_1v_2 \cdots \in \reach(F)$, and $\omega+1$ if $v_0 v_1 v_2 \cdots \notin \reach(F)$.
Given a strategy~$\sigma$ for Player~$0$, we define $\val_\game(v,\sigma) = \sup_\rho \val_\game(\rho)$, where $\rho$ ranges over all plays starting in $v$ that are consistent with $\sigma$.
Using König's Lemma shows that $\sigma$ is a winning strategy for Player~$0$ in $(\arena, \reach(F))$ from $v$ if and only if $\val_\game(v,\sigma) < \omega$ (here we use the assumption that $\arena$ is finitely branching).
Otherwise, i.e., if $\sigma$ is not winning from $v$, then we have $\val_\game(v, \sigma) = \omega+1$.
A strategy $\sigma$ for Player~$0$ is reachability optimal if $\val_\game(v,\sigma) \le \val_\game(v,\sigma')$ for every $v\in V$ and every strategy~$\sigma'$ for Player~$0$.
The existence of reachability optimal strategies follows straightforwardly from the correctness proof of the attractor construction for reachability games~\cite{GraedelThomasWilke02}. 

\begin{proposition}
Every reachability game in a finitely branching arena has a reachability optimal strategy.
\end{proposition}

Now, define $\arena' = ( V \cup E, V_1 \cup E, V_0, E', D )$ with 
\begin{itemize}
	\item $E' = \set{ (v,(v,v')),((v,v'),(v,v')) \mid (v,v') \in E }$ and
	\item $D = \set{((v,v'),v') \mid (v,v') \in E}$,
\end{itemize}
 i.e., we split every edge~$(v,v') \in E$ into a sequence of a standard edge from $v$ to the newly introduced vertex~$(v,v')$, a standard self-loop at $(v,v')$, and a disturbance edge from $(v,v')$ to $v'$.
Intuitively, it takes one disturbance in $\arena'$ to simulate a move in $\arena$ while the overall structure of the arena, including the strategic choices for the players, is preserved. 
Note that we flip the positions of Player~$0$ and $1$ when turning $\arena$ into $\arena'$.
Hence, we also dualize the winning condition and turn $\reach(F)$ into the dual safety condition~$\safety(F) = V^\omega \setminus \reach(F)$ and define $\game' = (\arena', \safety(F))$.
If $\game$ is a pushdown (one-counter) reachability game, then $\game'$ is a pushdown (one-counter) safety game and the blow-up of the transformation is polynomial.

Now, we relate the resilience of vertices in $\game'$ with the value of optimal strategies in~$\game$.

\begin{lemma}
\label{lemma_reachoptimality_vs_resilience}
Let $\game$ and $\game'$ as above. 
A reachability optimal strategy~$\sigma$ satisfies $\val_\game(v,\sigma) = r_{\game'}(v)$ for every $v \in V$.
\end{lemma}

\begin{proof}
Fix a reachability optimal strategy~$\sigma_\opt$ for Player~$0$ in $\game$.
We sketch the proof ideas, but leave the straightforward details to the reader.

{\boldmath$\val_\game(v,\sigma_\opt) \le r_{\game'}(v)$:}
The statement is trivial if $r_{\game'}(v) = \omega +1$. 
Hence, assume $r_{\game'}(v) \in \omega$, say $r_{\game'}(v) = k$.
Then, Player~$1$ has a winning strategy~$\tau$ for $(\arena'_\rig, \safety(F)_\rig \cup R_{\ge k+1})$ from $v$ due Lemma~\ref{lemma_riggedgamescharacterizeresilience}.\ref{lemma_riggedgamescharacterizeresilience_k}.
This strategy can be turned into a strategy~$\sigma$ for Player~$0$ in $\arena$ that mimics $\tau$ while ignoring the auxiliary vertices in $\arena'_\rig$ that are not in $\arena$. 
An induction shows that $F$ is reached within $k$ moves when starting in $v$ and playing according to $\sigma$.
Thus, $\val(v,\sigma_\opt) \le \val(v,\sigma) \le k = r_{\game'}(v)$.

{\boldmath$r_{\game'}(v) \le \val_\game(v,\sigma_\opt)$:}
The statement is trivial, if $\val_\game(v,\sigma_\opt) = \omega +1$. 
Hence, assume $\val_\game(v,\sigma_\opt)  \in \omega$, say $\val_\game(v,\sigma_\opt) = k$.
Then, by definition, $F$ is reached within $k$ moves when starting in $v$ and playing according to  $\sigma_\opt$.
The strategy~$\sigma_\opt$ for Player~$0$ in $\game$ can be turned into a strategy~$\tau$ for Player~$1$ in $\arena'_\rig$ that simulates the moves of $\sigma_\opt$ and, as long as $F$ has not been visited, simulate a disturbance whenever possible (there is a unique disturbance edge at every vertex with outgoing disturbances edges).
After visiting $F$, no more disturbances are simulated.
An induction shows that $F$ is reached and at most $k$ disturbances are simulated when starting in $v$ and playing according to $\tau$.
Hence, Player~$1$ wins 
\[(\arena'_\rig, \safety(F)_\rig \cup R_{\ge k+1}),\] which implies $r_\game(v) \le k = \val_\game(v,\sigma_\opt)$ by Lemma~\ref{lemma_riggedgamescharacterizeresilience}.\ref{lemma_riggedgamescharacterizeresilience_k}.
\end{proof}

Using the connection between resilience and values of reachability optimal strategies allows us to compute the value of a reachability optimal strategy in the initial vertex of a pushdown game. 
In particular, for one-counter systems, we obtain an algorithm with polynomial space requirements. 
Thereby, we close a gap in our knowledge about reachability optimal strategies in pushdown games.

\begin{theorem}
The following problem can be solved in polynomial space: \myquot{Given a one-counter reachability game~$\game$ with initial vertex~$v_\initmark$, determine $\val_\game(v_\initmark, \sigma)$ for a reachability optimal strategy~$\sigma$}.
\end{theorem}

Note that the approach via a reduction to computing the resilience presented here is not the simplest one:
One could simplify the constructions presented in Section~\ref{sec_ocs} and obtain a direct algorithm.

\section{Conclusion}
\label{sec_conclusion}
In this work, we have investigated pushdown safety games with disturbances, thereby extending the theory of games with disturbances from finite to infinite arenas.
In particular, we have determined the possible resilience values in safety games, presented effective characterizations for all possible values, and presented  algorithms that determine the resilience of the initial vertex (and a witnessing strategy) in one-counter and pushdown safety games.
As an application of our results, we obtained a polynomial space algorithm for computing optimal winning strategies for one-counter reachability games.
This is, to the best of our knowledge, the first improvement over the general doubly-exponential time algorithm for pushdown reachability games due to Carayol and Hague~\cite{CarayolH18}.

The algorithm computing the resilience in one-counter safety games runs in polynomial space, which is optimal, as the corresponding decision problems are $\pspace$-complete.
However, the algorithm for pushdown games has triply-ex\-ponential running time.
Here, there is a gap, as some of the corresponding decision problems are $\exptime$-complete (e.g., those for resilience $\omega+1$ and $\omega$) while the complexity of others is open (e.g., that for finite resilience values).
In future work, we aim to close this gap.
An interesting first step in that direction would be to determine the complexity of checking whether the resilience of the initial vertex is at least $k$, where $k$ is part of the input and encoded in binary.
Here, one has to keep in mind that algorithms for computing the resilience also yield algorithms computing optimal strategies in reachability games.
The latter problem also has a complexity gap between the currently best algorithms and known lower bounds.
Finally, another obvious open problem is to consider more general winning conditions, e.g., reachability or parity.

The main obstacle is that one either has to develop an effective characterization of vertices with resilience $\omega$ without a uniform witness, or to obtain an upper bound on the finite resilience value an initial vertex can assume. 
The first option is challenging due to the quantifier change discussed in Section~\ref{sec_pushdown}.
Hence, the more promising route seems to be the second option. 
The main challenge here is to bound the number of disturbances that are necessary to prevent Player~$0$ from ever reaching the target states, i.e., Player~$1$ now has a safety objective in conjunction with a limited number of disturbances at his disposal. 
\bibliographystyle{plain}
\bibliography{autocleaned}

\end{document}